\documentclass[10pt,a4paper,onecolumn,accepted=2021-03-03]{quantumarticle}
\pdfoutput=1

\usepackage[latin1]{inputenc}
\usepackage{amsmath, amscd}
\usepackage{amsfonts}
\usepackage{amssymb}
\usepackage{amsthm}
\usepackage{tikz-cd}
\usepackage{float}
\usepackage{bbm}
\usepackage{hyperref, cleveref}
\usepackage{doi}
\usepackage[numbers, sort&compress]{natbib}

\usepackage{url}

\usepackage{xcolor}

\usepackage{graphicx}
\graphicspath{{./figs/}}

\usepackage[margin=1in]{geometry}
\theoremstyle{plain}
\newtheorem{thm}{Theorem}[section]
\newtheorem{lem}[thm]{Lemma}
\newtheorem{prop}[thm]{Proposition}
\newtheorem{cor}[thm]{Corollary}

\theoremstyle{definition}

\newcommand{\Z}{\mathbb{Z}}
\newcommand{\A}{\mathcal{A}}

\newcommand{\C}{\mathbb{C}}

\newcommand{\M}{\mathcal{M}}

\newcommand{\Tr}{\textrm{Tr}}

\newcommand{\Irr}{\textrm{Irr}}
\newcommand{\Span}{\textrm{span}}
\newcommand{\Sym}{\textrm{Sym}}
\newcommand{\im}{\textrm{Im}}
\newcommand{\Fun}{\textrm{Fun}}
\newcommand{\End}{\textrm{End}}
\newcommand{\wt}{\textrm{ wt}}
\newcommand{\F}{\mathbb{F}}

\title{Quantum query complexity of symmetric oracle problems.}
\author{Daniel Copeland}
\email{daniel.copeland@gmail.com} 
\affiliation{UC San Diego}

\author{Jamie Pommersheim}
\email{jamie@reed.edu}
\affiliation{Reed College}
\begin{document}

\maketitle

\begin{abstract}
We study the query complexity of quantum learning problems in which the oracles form a group $G$ of unitary matrices. In the simplest case, one wishes to identify the oracle, and we find a description of the optimal success probability of a $t$-query quantum algorithm in terms of group characters.  As an application, we show that $\Omega(n)$ queries are required to identify a random permutation in $S_n$. More generally, suppose $H$ is a fixed subgroup of the group $G$ of oracles, and given access to an oracle sampled uniformly from $G$, we want to learn which coset of $H$ the oracle belongs to. We call this problem coset identification and it generalizes a number of well-known quantum algorithms including the Bernstein-Vazirani problem, the van Dam problem and finite field polynomial interpolation. We provide character-theoretic formulas for the optimal success probability achieved by a $t$-query algorithm for this problem. One application involves the Heisenberg group and provides a family of problems depending on $n$ which require $n+1$ queries classically and only $1$ query quantumly.
\end{abstract}



\section{Introduction}
An oracle problem is a learning task in which a learner tries to determine some information by asking certain questions to a teacher, called an oracle. In our setting the learner is a quantum computer and the oracle is an unknown unitary operator acting on some subsystem of the computer. The computer asks questions by preparing states, subjecting them to the oracle, measuring the results, and finally making a guess about the hidden information. How many queries to the oracle are needed by the computer to guess the correct answer with high probability?\\

This paper addresses the following oracle problem. Fix a finite group $G$ and a subgroup $H \leq G$. The elements of $G$ are encoded as unitary operators by some unitary representation $\pi: G \to U(V)$. Given oracle access to $\pi(a)$ (for some unknown $a \in G$) the learner must guess which coset of $H$ the element $a$ lies in.  We focus on average case success probability, though an easy averaging argument, given in Section \ref{sec:2}, shows that the worst case and average case query complexity are equal.\\

We call this problem {\it coset identification}. This task encompasses many of previously studied qauntum oracle problems, including univariate and multivariate polynomial interpolation over a finite field \cite{CvDHS:2016, CCH:2018}, the group summation problem \cite{MePo:2011, zhandry:2015, BBCMW:qlbpoly}, and symmetric oracle discrimination \cite{BCMP:sod}.  In addition, the coset identification problem we study generalizes the homomorphism evaluation problem for abelian groups studied by Zhandry in \cite{zhandry:2015}, which greatly inspired us.  Section \ref{sec:7} gives details of this connection. \\

In this paper, we analyze the query complexity of the general coset identification problem.  We prove that nonadaptive algorithms are optimal for any coset identification problem. We provide tools to reduce the analysis of query complexity to purely character theoretic questions (which are themselves often combinatorial). In particular we derive a formula for the exact quantum query complexity for coset identification\ in terms of characters. In the case of symmetric oracle discrimination (which itself includes polynomial interpolation as a special case) we find the lower and upper bound for bounded error query complexity. \\

Another motivation for our work is the study of nonabelian oracles. Much is known about quantum speedups when the oracle is a standard Boolean oracle. Less is known about whether oracle problems with nonabelian symmetries can offer notable speedups. To that end we study the follow scenario: suppose a group $G$ acts by permutations on a finite set $\Omega$ (we call $\Omega$ a $G$-set). A learner is given access to a machine which takes an element $\omega \in \Omega$ and returns $a \cdot \omega$ for some hidden group element $a \in G$. With as few queries as possible the learner should guess the hidden element $a \in G$. The classical query complexity for this problem is a long-known invariant of $G$-sets called the base size. For instance, if $G$ is the full permutation group of $\Omega = \{1, \dots, n\}$ then $n-1$ queries are required classically to determine the hidden permutation. This problem is a special case of symmetric oracle discrimination and we can express the bounded error quantum query complexity of this purely in terms of the character of the $G$-set $\Omega$. For instance, we find that when $G$ is the full permutation group of $X = \{1, \dots, n\}$ then $n - 2\sqrt{n} + \Theta(n^{1/6})$ queries are necessary (and sufficient) to determine the hidden element. \\

This result bears some similarity to other work on learning problems related to the symmetric group. Aaronson and Ambainis \cite{AA2006}, who prove that at most a polynomial speedup can be achieved in computing functions on $n$ inputs which are invariant under the action of the full symmetric group $S_n$ (using a standard evaluation oracle). Ben-David \cite{BenDavid2016} proves that at most a polynomial speedup is possible for Boolean functions defined on the full symmetric group. More recently, Dafni, Filmus, Lifshitz, Lindzay and Vinyals \cite{DFLLV2021} have studied the query complexity of Boolean functions defined on the symmetric group, again proving a polynomial relationship between the quantum and classical query complexities (as well as numerous other complexity measures). These results may be compared to the well-known fact that only polynomial speedups are possible in computing total Boolean functions \cite{BBCMW:qlbpoly}, the idea being that learning problems on the full symmetric group correspond to total functions, while learning problems on a subgroup correspond to partial functions. All of the results mentioned above are not directly comparable to ours, since they use a standard evaluation oracle, while we examine a more symmetric ``in-place" oracle model. \\

The task of oracle identification can be further refined: fix a group $G$, a $G$-set $\Omega$, and a function $f: G \to X$ which is constant on left cosets of some subgroup $H$, and distinct on distinct cosets. The {\it (left) coset identification problem} is to determine $f(a)$ given access to a permutational black-box hiding $a$ through the action on $\Omega$. For instance, when $G = S_n$ (the symmetric group), $\Omega = \{1, \dots, n\}$ its defining representation and $f$ the sign homomorphism, it requires $n-1$ classical queries to determine $f(a)$. As a counterpoint to the harsh lower bound above we provide a family of examples for this task parametrized by $n$ in which the quantum query complexity is $1$ while the classical complexity is $O(n)$. The groups we use are Heisenberg groups acting as small subgroups of the full permutation group. This example is a nonabelian analogue of the fact that good quantum speedups can be found in computing partial Boolean functions \cite{BV:1997}. \\

The paper is organized as follows. In \cref{sec:2} we formalize coset identification\ in the context of quantum learning algorithms and review the notions of adaptive and nonadaptive learning. In \cref{sec:3} we prove that parallel queries suffice to produce an optimal algorithm for this task. \Cref{sec:4} applies this theorem to symmetric oracle discrimination and addresses numerous example problems. In \cref{sec:5} we return to the general coset identification task and we prove the main theorem of this paper, Theorem \ref{CIThm}, which is a formula for the success probability of an optimal $t$-query algorithm in terms of characters. We use this in \cref{sec:6} to compute the exact and bounded error query complexity of some special examples (including the Heisenberg group example). We conclude in \cref{sec:7} by explaining how our work reproduces several previously known results involving abelian oracles. \\

Our paper uses the language of representation theory of finite groups. A suitable reference is the first third of Serre's textbook \cite{serre:linear}. We review some important notations later in Section 5 (in particular, the idea of induced representation is critical for the statement of our results.) Here we mention that a {\it representation} of a finite group $G$ always refers to a {\it finite dimensional} and {\it unitary} representation of $G$ over the complex numbers. In other words, a representation is a group homomorphism $\pi: G \to U(V)$ (the unitary group of a f.d. vector space $V$). We often think of $V$ as a left module for the group alegbra $\C G$, and use the notation $gv$ for $\pi(g)v$ when the map $\pi$ is clear from the context.

\section{Quantum learning from oracles}
\label{sec:2}
A quantum or classical oracle problem is described by a set of hidden information $Y$, a function $f: Y \to X$ (the function to learn or compute), and a representation of $Y$ as operations on inputs of some kind (which determines the oracles). Classically such a representation consists of a set of inputs $\Omega$ and an assignment taking each $y \in Y$ to a permutation of $\Omega$, i.e. a map $\pi: Y \to \Sym(\Omega)$. A {\it classical oracle problem} is specified by a tuple $(Y, \Omega, \pi, f)$. A classical computer has access to $\pi(y)$ for some unknown $y \in Y$ by spending one query to input $\omega \in \Omega$ to learn $\pi(y) \cdot \omega$. The goal is to determine $f(y)$ with a high degree of certainty with as few queries as possible. More concretely, we measure the efficacy of an algorithm by its average case success probability, namely the probability of correctly outputting $f(y)$ supposing the hidden information $y$ is sampled uniformly from $Y$. For the highly symmetric problems considered in this paper, this is the same as the worst-case success probability, as explained below. \\

The quantum representation of oracles is described by a Hilbert space $V$ and an assignment taking each $y \in Y$ to a unitary operator of $V$, in other words a map $\pi: Y \to U(V)$. Thus a {\it quantum oracle problem} is specified by a tuple $(Y, V, \pi, f)$. The quantum computer spends one query to input a state $|\psi \rangle \in V$ to $\pi(y)$ to acquire the state $\pi(y) |\psi \rangle$; the goal is to produce a state and measurement scheme which outputs the value $f(y)$. \\

Any classical oracle problem $(Y, \Omega, \pi, f)$ determines a quantum oracle problem via linearization: oracles will act on the Hilbert space $\C \Omega$ (spanned by the orthonormal basis $\{|\omega \rangle\ |\ \omega \in \Omega\}$) by permutation matrices. \\

We note that there are other oracle models used to encode permutations. One possibility is to require an oracle to act on a bipartite system, with one subsystem specified to be the ``input register" and the other a ``response register". \footnote{More precisely, one usually defines an abelian group structure on $\Omega$ (usually cyclic) by defining an operation $\oplus$ on $\Omega$. Then the oracle hiding the permutation $\pi(a)$ is defined to act by  $|\omega, b \rangle \mapsto |\omega, (\pi(a) \cdot \omega) \oplus b \rangle$. Here $\omega, b \in \Omega$, so both the input and response registers are copies of $\C \Omega$.} 
While we do not specifically consider this model here, we note that many oracle problems, such as polynomial interpolation and group summation, that are normally formulated in this two-register setup do have an easy reformulation in our setup.    Thus, our results and analyses apply to these problems in their original two-register formulation.  See Section 7.
However, 
in some cases, the two-register setup results in a set of oracles that do not form a group, for instance in  the work of Ambainis on permutation inversion \cite{ambainis:2002}. In general, it is an interesting question (and to our knowledge, open) whether these oracle models are the same, or if they lead to different query complexities. \footnote{As another modification, one may propose that having access to an oracle $\mathcal{O}$ means an algorithm may choose to access $\mathcal{O}$ or $\mathcal{O}^{-1}$ in any given query. This is a separate model which we do not consider here.}
 \\

A {\it symmetric oracle problem} is an oracle problem in which the hidden information is a group $G$ (so we are replacing $Y$ with $G$) and the map $\pi$ is a homomorphism $G\rightarrow \Sym(\Omega)$ in the classical case or  $G\rightarrow U(V)$ in the quantum case.  If $\pi: G\rightarrow U(V)$ is a homomorphism, then it is common practice to regard $V$ as a (left) $\C G$-module where $\C G$ is the group algebra of $G$ (spanned by an orthonormal basis sometimes written without kets as $\{g\ |\ g \in G\}$. In module notation we sometimes write $g\cdot v := \pi(g)(v)$ (for $g \in G, v \in V$) if the representation $\pi$ is understood from context. The quantum oracle arising from a symmetric classical problem is also symmetric. \\

Of special interest to us is the case when the function $f$ to be learned is compatible with the group structure $G$. An instance of the {\it coset identification problem} is a symmetric oracle problem $(G, V, \pi,  f)$ where the function $f: G \to X$ is constant on left cosets of a subgroup $H \leq G$ and distinct on distinct cosets. We also assume $f$ is onto. The typical example is when $X = \{gH\ |\ g \in G\}$ is the set of left cosets of $H$ and $f(g) = gH$. An equivalent formulation is to say that $X$ is a transitive $G$-set and the map $f: G \to X$ is a map of (left) $G$-sets (i.e., $f(gh) = gf(h)$ for all $g, h \in G$). Then the subgroup $H$ can be recovered as the preimage of $f(e)$. \\

For our analysis of the coset identification problem, we focus on average case success probability. The symmetry of the problem implies that worst case and average case success probabilites are equal, as the following argument shows:  provided an unknown oracle $\pi(a)$ we can select $g \in G$ uniformly at random and preprocess our input by applying $\pi(g)$. Then an optimal average-case algorithm will return the coset containing $ga$ with optimal average-case success probability. The coset which contains $a$ can then be retrieved by applying $g^{-1}$. Hence it suffices to consider the average case success probability of any algorithm for this task (with the unknown oracle $\pi(a)$ sampled uniformly from $G$).\\

We examine bounded error and exact measures of query complexity. The {\it exact} (or {\it zero error}) {\it query complexity} of a learning problem is the minimum number of queries needed by an algorithm to compute $f(y)$ with zero probability of error. The {\it bounded error query complexity} is the minimum number of queries needed by an algorithm to compute $f(y)$ with probability $\geq 2/3$. The bounded error query complexity is often studied for a family of problems growing with a parameter $n$ and so changing the constant $2/3$ above to any number strictly greater than $1/2$ will only change the query complexity by a constant factor mostly ignored in asymptotic analysis. \\

Broadly speaking, there are two qualitatively different approaches to solving an oracle problem. The first approach is to ask questions one at a time, carefully changing your questions as you receive more information. This is called using adaptive queries. The other approach is to prepare all your questions and ask them at once in one go (imagining the learner has access to multiple copies of the teacher). This is known as using non-adaptive, or parallel queries.\\

Classically the adaptive model is at least as strong as the nonadaptive model, since you can convert any nonadaptive algorithm into an adaptive one (by picking your questions in advance but asking them one at a time). This is well-known to be true also in the quantum setting. In the next section we will prove the converse for coset identification: \\

\begin{thm}\label{parallelqueries}
Suppose $(G, V, \pi, f)$ describes an instance of coset identification. Then there exists a $t$-query quantum algorithm to determine $f(a)$ with probability $P$ if and only if there exists a $t$-query nonadaptive query algorithm which does the same.
\end{thm}

This theorem is certainly not true for arbitrary learning problems: Grover's algorithm provides an example in which any optimal algorithm must use adaptive queries \cite{zalka:1999}.  To prove the theorem we must precisely state what adaptive and nonadaptive algorithms are. 

\subsection{Adaptive vs. nonadaptive: definitions}\label{sec:algdefs}
Recall that a quantum learning problem is described by a tuple $(Y,V, \pi, f: Y \to X)$ where $Y$ indexes the set of hidden information, $V$ is a finite dimensional Hilbert space, $\pi: Y \to U(V)$ a representation of the unknown information by unitary operators, and $f$ is the function to learn.\\

The standard model for an adaptive algorithm is as follows (see {\it e.g.} \cite[Section 3.2]{BBCMW:qlbpoly}): \\

A {\it $t$-query adaptive quantum algorithm} for the quantum oracle problem $(Y, V, \pi, f: Y \to X)$ consists of a tuple $\A = (N, \psi, \{U_1, \dots, U_t\}, \{E_i\})$ where
\begin{itemize} \item{$N$ is the dimension of the auxiliary workspace used in the computation}
\item{$|\psi\rangle$ is a unit vector in $V \otimes \mathbb{C}^N$}
\item{$\{U_1, \dots, U_t \}$ is a set of unitary operators acting on $V \otimes \mathbb{C}^N$}
\item{$\{E_x\}_{x \in X}$ is a POVM with measurement outcomes indexed by $X$.}\\
\end{itemize}

The algorithm uses $t$ queries to the oracle $\pi(a)$ (with $a$ sampled uniformly from $Y$) to produce the output state
$$
|\psi_a^\A \rangle =U_t (\pi(a) \otimes I) U_{t -1} (\pi(a) \otimes I) \dots (\pi(a) \otimes I) U_1 (\pi(a) \otimes I) |\psi \rangle
$$
upon which the algorithm executes the measurement described by $\{E_x\}_{x \in X}$. Here and elsewhere $I$ denotes the identity operator (in this case acting on the space $\mathbb{C}^N$). \\

In quantum circuit notation the preparation of the state $|\psi_a^{\mathcal{A}}\rangle$ reads:
\begin{figure}[h!]\label{AdaptiveCircuit}
\centering
\includegraphics[scale=1.0]{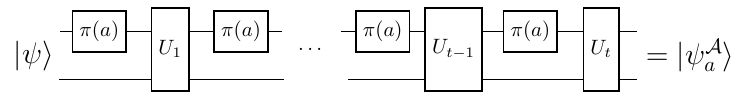}
\end{figure}

By contrast, an algorithm is {\it nonadaptive} if at any point during the algorithm, the input for some query does not depend on the results to any of the previous queries. Essentially this means that all the inputs are completely determined before the algorithm begins. Classically, $t$ nonadaptive queries are identical to $t$ simultaneous queries to $t$ copies of an oracle. This motivates the following definition ({\it cf}\  \cite[Section 2]{mont:2010}):\\

A {\it $t$-query nonadaptive quantum algorithm} for the oracle problem $(Y, V, \pi, f)$ is a tuple $\A = (N, \psi, \{E_x\})_{x \in X}$ where
\begin{itemize}
\item{$N$ is the dimension of the auxiliary register.}
\item{$| \psi \rangle$ is the input state, a unit vector of $V^{\otimes t} \otimes \C^N$.}
\item{$\{E_x\}$ is a POVM indexed by $X$}. \\
\end{itemize}

The algorithm operates on the input state to produce
$$
|\psi_a^{\A} \rangle = (\pi(a)^{\otimes t} \otimes I) |\psi \rangle
$$
which is then measured using the POVM $\{E_x\}$. The next fact is very useful and follows immediately from definitions.

\begin{lem}\label{tquery}
A $t$-query nonadaptive algorithm for the problem $(Y, V, \pi, f)$ is the same as a single-query nonadaptive algorithm for the oracle problem $(Y, V^{\otimes t}, \pi^{\otimes t}, f)$. 
\end{lem}
The quantum circuit notation for the nonadaptive preparation of the state $|\psi_a^\A \rangle$ is drawn as follows. \\

\begin{figure}[h!]
\centering
\includegraphics[scale=1.0]{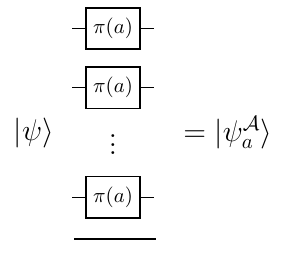}
\end{figure}

In either model, the algorithm $\A$ uses $t$ copies of the unitary $\pi(a)$ to produce a state $|\psi_a^A \rangle$. Using the POVM $\{E_x\}$ results in a measurement value $x \in X$ with probability 

$$P(x\ |\ a) = \langle \psi_a^\A | E_x | \psi_a^\A \rangle.$$

Since we assume the oracle is sampled uniformly from $Y$, the probability that $\A$ executes successfully is
$$
P_{\rm succ}(\A) = \frac{1}{|Y|} \sum_{a \in Y} P(f(a)\ |\ a) = \frac{1}{|Y|}\sum_{a \in Y} \langle \psi_a^\A | E_{f(a)} | \psi_a^\A \rangle.
$$ 


\subsection{Symmetric oracle problems}
Suppose we have a symmetric oracle problem $(G, V, \pi, f)$. As mentioned in the introduction, since we are focusing on query complexity and not on issues of implementation, analysis of this problem depends only on the character $\chi_V$ of $\pi: G \to U(V)$, as we prove in the lemma below. In fact, a little more is true. Let $\Irr(G)$ denote the set of irreducible characters of $G$. Given a representation $\pi: G \to U(V)$ define the set
\begin{align*}
I(V) &:= \{\chi \in \Irr(G) \text{ appearing in the representation } V\} \\
&= \{\chi \in \Irr(G)\ |\ (\chi, \chi_V) > 0\}.
\end{align*}
Here we are using $(\cdot,\cdot)$ to denote the usual inner product of characters. If $\chi \in \Irr(G)$ and $(\chi, \chi_V) > 0$ we say that $\chi$ appears in the representation $V$.

\begin{lem}\label{RepDependence} The optimal success probability of a $t$-query algorithm to solve a symmetric oracle problem $(G, V, \pi, f)$ depends only on $I(V)$ and $f$.
\end{lem}
\begin{proof}
First, note that if $U: V \to W$ is a Hilbert space isomorphism then we can define a new oracle problem $(G, W, U\pi U^{-1}, f)$ where the oracles now act on $W$. Any $t$-query algorithm to solve the original problem can be ``conjugated" by $U$ (e.g. the input state $|\psi \rangle$ becomes $U|\psi \rangle$ and the non-oracle unitaries and POVM are conjugated by $U$) to produce a $t$-query algorithm for the new problem which succeeds with the same probability. Conversely any algorithm to solve the new problem can be conjugated by $U^{-1}$ to solve the old problem with the same probability. Therefore oracle problems with isomorphic unitary representations of $G$ will have the same $t$-query optimal success probability. In other words, only the character $\chi_V$ is relevant.\\

Second, we claim that the multiplicities of irreducible characters in $V$ are not important; only whether they appear in $V$ or not. Indeed, adding a $d$-dimensional workspace to a computer's original system $V$ produces a new representation $V \otimes \C^d$ of $G$ with character $d\chi_V$. Since we allow our algorithm to introduce any such workspace, we are in effect allowing it to increase the multiplicity of each character by a factor of $d$. Note that this process will never produce irreps which did not appear in $V$ to begin with. Hence the optimal success probability depends only on which irreps appear in $V$, i.e. the set $I(V)$.
\end{proof}

It makes sense that if an algorithm is granted access to more representations to work with, its success probability cannot decrease. To be more precise, fix $t$, and let $P_{\rm opt}(G, V, \pi, f)$ denote the optimal success probability of a $t$-query algorithm for the symmetric  oracle problem $(G, V, \pi, f)$.
\begin{lem}\label{lem:monotone}
Suppose $\pi_V$, $\pi_W$ are representations of $G$ on the spaces $V$ and $W$, with $I(W) \subset I(V)$. Then
$$
P_{\rm opt}(G, W, \pi_W, f) \leq P_{\rm opt}(G, V, \pi_V, f).
$$
\end{lem}
\begin{proof}
The basic idea is any $t$-query algorithm to solve $(G, W, \pi_W, f)$ can be extended to produce a $t$-query algorithm for $(G,V, \pi_V, f)$. Suppose an algorithm $\A$ for $W$ uses an $N$ dimensional ancilla space, i.e. operates on $W \otimes \C^N$. Since $I(W) \subset I(V)$, there exists some $M$ so that $V\otimes \C^M$ contains a subrepresentation isomorphic to $W \otimes \C^N$. Hence we can write $V \otimes \C^M = W' \oplus Y$ where $W' \cong W \otimes \C^N$ as $\C G$-modules. Now we claim the initial state, intermediate unitaries, and POVM for the algorithm $\A$ can be extended to an algorithm $\A'$ acting on $V \otimes \C^M$. The initial state for $\A'$ is the vector in $W' \subset V \otimes \C^N$ corresponding to the initial state for $\A$ in $W \otimes \C^N$. The intermediate unitaries for $\A'$ act on $W'$ according to the unitaries for the algorithm $\A$ and are extended arbitrarily to $Y$. The measurement operators for $\A'$ all agree with the measurement operators for $\A$ on $W'$ and all but one of the operators act as $0$ on the subspace $Y$. To satisfy the completeness relation on $V \otimes \C^M$, exactly one of the POVMs should act as the identity on $Y$ (this modification is unimportant since $\A'$ ``takes place" entirely within $W'$). The success probability of $\A'$ is equal to that of $\A$.
\end{proof}


\section{Parallel queries suffice}
\label{sec:3}
Here we prove Theorem \ref{parallelqueries}, namely that the optimal success probability for coset identification\ can be attained by a parallel (nonadaptive) algorithm. We prove this by showing that any $t$-query adaptive algorithm can be converted to a $t$-query nonadaptive algorithm without affecting the success probability. Another way to say this is that every $t$-query adaptive algorithm can be simulated by a $t$-query nonadaptive one. This technique is greatly inspired by the work of Zhandry \cite{zhandry:2015} who proves this result when $G$ is abelian, and also bears resemblance to the lower bound technique of Childs, van Dam, Hung and Shparlinski \cite{CvDHS:2016}, where the special case of polynomial interpolation is addressed. \\

Let $\pi: G \to U(V)$ be a unitary representation of $G$. Let $\C G$ denote the group algebra of $G$. Each $h \in G$ acts on $\C G$ by left multiplication, an operator we denote $L_h$. We will use the {\it controlled multiplication} operator (\cite{dBCW:2002}) defined on $V \otimes \C G$ by
$$
CM|v, g \rangle = |\pi(g^{-1})v, g \rangle.
$$
This defines a unitary operator and is a generalization of the standard CNOT gate (take $G = \Z_2$ and $V = \C\Z_2$). As such we draw it using circuit diagrams as in \cref{CMgate}.
\begin{figure}[h!]
\label{CMgate}
\centering
\includegraphics[scale=.5]{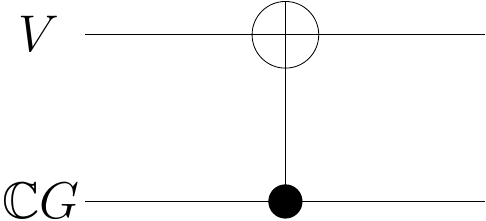}
\caption{Notation for the controlled multiplication gate $CM$.}
\end{figure}\\
There are two $G$-actions on $V \otimes \C G$ we use, one given by $\pi(h) \otimes L_h$ and the other $I \otimes L_h$. Our first observation is that $CM$ intertwines these actions.
\begin{lem}\label{CMlemma}
The controlled multiplication operator satisfies $$CM(\pi(h) \otimes L_h) = (I \otimes L_h)CM.
$$
\end{lem}
  
The proof follows by applying both sides to a vector $|v, g \rangle$ and using the definition of $CM$. The representation obtained by letting each $h \in G$ act by the identity on $V$ is a direct sum of $\dim V$ many copies of the trivial reprseentation, so we denote it $\mathbbm{1}^{\oplus \dim V}$. The lemma allows us to interpret $CM$ as a $\C G$-module isomorphism $V \otimes \C G \to \mathbbm{1}^{\oplus \dim V} \otimes \C G$. In pictures the lemma reads:
\begin{figure}[h!]
\centering
\includegraphics[scale=.5]{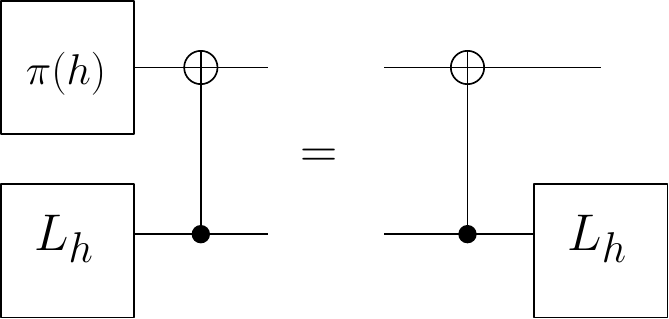}
\caption{Lemma \ref{CMlemma} in pictures.}
\end{figure}\\

The next property is crucial for our parallelization argument. Recall that if $W$ is a $\C G$-module then $I(W)$ denotes the set of irreducible characters of $G$ which appear in $W$.

\begin{lem}\label{CMlemma2}
Suppose $W$ is a subrepresentation of $\C G$. Then there is a subrepresentation $Y$ of $\C G$ such that the image of $V \otimes W$ under $CM$ is contained in $V \otimes Y$ and $Y$ satisfies $I(Y) = I(V \otimes W)$.
\end{lem}

\begin{proof} By Lemma \ref{CMlemma} $CM$ is a $\C G$-module isomorphism $V \otimes \C G \to \mathbbm{1}^{\oplus \dim V} \otimes \C G$ where $V$ and $\mathbbm{1}^{\oplus \dim V}$ have the same underlying vector space. Let $Z$ denote the image of $V \otimes W$ under $CM$. Then $CM$ restricts to a $\C G$-module isomorphism $V \otimes W \to Z$. Next let $Y$ be the submodule of $\C G$ which contains each irreducible of $I(Z)$ with maximal multiplicity (so if $\chi$ appears in $Y$ then $\chi$ appears with multiplicity $\chi(e)$). Now $Z \cong V \otimes W$ as $\C G$-modules so in particular $I(Z) = I(V \otimes W)$. Hence also $I(Y) = I(V \otimes W)$. 

It remains to prove $Z \subseteq V \otimes Y$. Indeed, in the $\C G$-module $\mathbbm{1}^{\oplus \dim V} \otimes \C G$ the subspace $\mathbbm{1}^{\oplus \dim V} \otimes Y$ is the maximal subrepresentation containing only irreducibles in $I(V \otimes W)$. As noted $Z$ contains only irreducibles in $I(V \otimes W)$ so therefore $Z \subseteq \mathbbm{1}^{\oplus \dim V} \otimes Y$, which is the same vector space as $V \otimes Y$.
\end{proof}

\sloppy Now suppose $(G, V, \pi, f)$ is an instance of coset identification\ and $\A = (N, |\psi \rangle, \{U_1, \dots, U_t\}, \{E_x\} )$ is a $t$-query adaptive algorithm to evaluate the homomorphism $f$. First, by replacing $\pi$ with $\pi \otimes I$ if necessary, we may assume that the algorithm does not use a workspace, that is $N = 1$. We will describe a new adaptive algorithm $\A'$ which is a modification of $\A$ as follows. We introduce a new workspace which is a copy of $\C G$. The new intermediate unitaries are $(U_1 \otimes I) CM,(U_2 \otimes I)CM, \dots, (U_t \otimes I) CM$. The input state is $|\psi \rangle \otimes |\eta \rangle$ where $\eta$ is the equal superposition state in $\C G$. When the oracle is hiding the unitary $\pi(a)$ this produces the following state: \\
\begin{figure}[h!]
\centering
\includegraphics[scale=1.1]{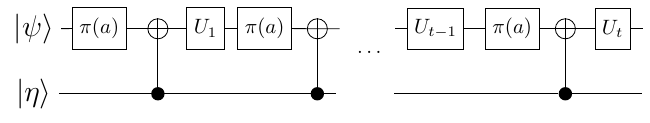}
\caption{Pre-measurement state for $\A'$.}\label{ModifiedAlg1}
\end{figure} \\

Next measurement is performed: first the second register is measured in the standard basis of $\C G$. Then the original POVM is applied to the first register. The result of these two measurements will be a pair $(g, x)$; the final output of the algorithm is $gx$. \footnote{Formally the algorithm $\A'$ is given by
$$
\A' = (|G|, |\psi, \eta \rangle, \{U_1 \otimes I \circ CM, \dots, U_t \otimes I \circ CM\}, \{E'_x = \sum_{g \in G}E_{g^{-1}x} \otimes |g \rangle \langle g| \}).
$$}
\begin{lem}\label{lemmaA}
The algorithm $\A'$ succeeds with the same success probability as $\A$.
\end{lem}
\begin{lem}\label{lemmaB}
The algorithm $\A'$ can be simulated by a $t$-query parallel query algorithm.
\end{lem}

\begin{proof}[ Proof of Theorem \ref{parallelqueries} from Lemmas \ref{lemmaA} and \ref{lemmaB}]
By the two lemmas, given any $t$-query adaptive algorithm $\A$ which solves coset identification\ with probability $P$, there exists a $t$-query parallel query algorithm which succeeds with the same probability. \end{proof}

\begin{proof}[Proof of Lemma \ref{lemmaA}]
Consider the pre-measurement state for $\A'$ given that the hidden group element is $a \in G$. It can be written
$$
|\psi_a^{\A'} \rangle = \frac{1}{\sqrt{|G|}}\sum_{g \in G}| \psi_{g^{-1}a}^{\A} \rangle \otimes |g \rangle.
$$
If the first measurement reads $g$ then the state collapses to $| \psi^{\mathcal{A}}_{g^{-1}a} \rangle \otimes |g \rangle$. If the second measurement is now performed, the result will read $f(g^{-1}a)$ with the same probability that the algorithm $\A$ would read this result given that the oracle was hiding $g^{-1}a$. The algorithm then classically converts the result to $gf(g^{-1}a)$ which is equal to $f(a)$ since $f$ is a left $G$-set map. So the following conditional probabilities are equal:
\begin{align*}
P(\A' \text{ outputs } f(a)\ &|\ a \text{ is hidden }, \text{ first measurement result is } g) \\
 &= P(\A \text{ outputs } f(g^{-1}(a))\ |\ g^{-1}a \text{ is hidden }).
\end{align*}
Denote these probabilities by $P_{\A'}(f(a)\ |\ a, g)$ and $P_{\A}(f(g^{-1}a)\ |\ g^{-1}a)$ respectively. Since the probability that the first measurement of $\A'$ reads $g$ is $1/|G|$ for all $G$  and $g$ is sampled independently of $a$, we compute the average case success probability by
\begin{align*}
P_{\textrm{succ}}(\A') &= \frac{1}{|G|^2} \sum_{g \in G} \sum_{a \in G} P_{\A'}(f(a)\ |\ a, g) \\
&= \frac{1}{|G|^2}\sum_{g \in G} \sum_{a \in G} P_{\A}(f(g^{-1}a)\ |\ g^{-1}a)\\
&= \frac{1}{|G|}\sum_{g \in G} P_{\textrm{succ}}(\A) = P_{\textrm{succ}}(\A).
\end{align*} \end{proof}

\begin{proof}[Proof of Lemma \ref{lemmaB}] We rewrite the pre-measurement state of $\A'$ expressed by Figure \ref{ModifiedAlg1} using Lemma \ref{CMlemma}. Denote the state that results when the hidden element is $a \in G$ by $|\psi_a^{\A'} \rangle$. We apply Lemma \ref{CMlemma} diagrammatically from left to right: \\
\begin{center}
\includegraphics[scale=1.0]{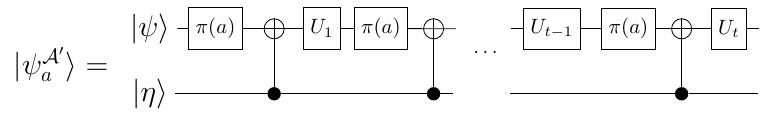}
\includegraphics[scale=1.0]{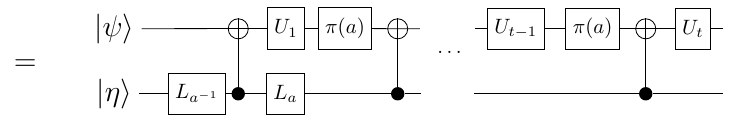}
\end{center}
\begin{center}
\includegraphics[scale=1.0]{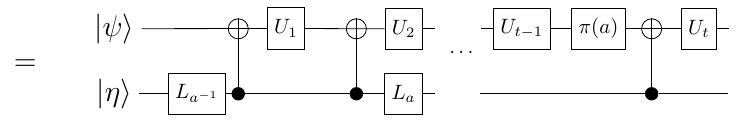}
\includegraphics[scale=1.0]{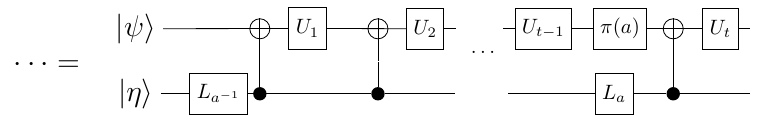}
\includegraphics[scale=1.0]{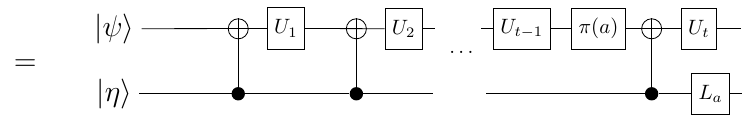}
\end{center}

In the last step, in addition to applying Lemma \ref{CMlemma} at the right of the diagram, we used the fact that $L_{a^{-1}}|\eta \rangle = | \eta \rangle$. In formulas we have
$$
|\psi_a^{\A'} \rangle = (I \otimes L_a) \circ \big((U_t \otimes I) \circ CM \circ \dots \circ (U_1 \otimes I) \circ CM \big)|\psi, \eta \rangle.
$$
Therefore we have converted this algorithm to a single-query algorithm using the oracle $I \otimes L_a$ with initial state $U|\psi, \eta \rangle$ where $U = (U_t \otimes I) \circ CM \circ \dots \circ (U_1 \otimes I) \circ CM$. \\

{\it Claim.} The image of $V \otimes \C|\eta \rangle$ under $U$ is contained in $V \otimes Y$ where $Y \subseteq \C G$ is a submodule satisfying $I(Y) = I(V^{\otimes t})$. \\

This is readily proved by induction and Lemma \ref{CMlemma2}. For instance, by Lemma \ref{CMlemma2} the image of $V \otimes \C |\eta \rangle$ under $CM$ is contained in $V \otimes Y_1$ where $Y_1$ is a submodule with $I(Y_1) = I(V)$. The next part of $U$ is $U_1 \otimes I$ which sends $V \otimes Y_1$ to itself. Now another $CM$ is applied and by Lemma \ref{CMlemma2} this sends $V \otimes Y_1$ to $V \otimes Y_2$ where $I(Y_2) = I(V \otimes Y_1) = I(V ^{\otimes 2})$. \\

Therefore the inital state $U|\psi, \eta \rangle$ belongs to the subspace $V \otimes Y$, which means that the algorithm $\A'$ may be simulated by a single query algorithm to the oracle $I \otimes L_a$ acting on the subspace $V \otimes Y$. Note that the irreducibles appearing in this subspace are $I(\mathbbm{1}^{\oplus \dim V} \otimes Y) = I(Y) = I(V^{\otimes t})$. Hence Lemma \ref{RepDependence} implies there exists a single-query algorithm using the representation $V^{\otimes t}$ which achieves the same success probability as $\A'$. As noted in Lemma \ref{tquery} this is the same as a $t$-query parallel algorithm using the representation $V$. This concludes the proof of Lemma \ref{lemmaB}. \end{proof} 

{

\begin{cor}\label{tQueryCor}
The optimal $t$-query success probability for an algorithm solving an instance of coset identification $(G, V, \pi, f)$ is equal to the optimal single-query success probability achievable solving the instance $(G, V^{\otimes t}, \pi^{\otimes t}, f)$. \qed
\end{cor}
}

\section{Application to symmetric oracle identification}
\label{sec:4}
Symmetric oracle discrimination is the following task: given oracle access to a symmetric oracle hiding a group element $a \in G$, determine $a$ exactly. This is the special case of coset identification in which $H = \{e\}$. Thus an instance of this problem is determined by a finite group $G$ and a (finite-dim) unitary representation $\pi: G \to U(V)$. The following theorem computes the success probability of a single-query algorithm and is proved by Bucicovschi, Copeland, Meyer and Pommersheim:

\begin{thm}\label{singlequerySOD} (\cite{BCMP:sod}, Theorem 1) Suppose $G$ is a finite group and $\pi: G \to U(V)$ a unitary representation of $G$. Then an optimal single-query algorithm to solve symmetric oracle discrimination succeeds with probability
$$
P_{\rm opt} = \frac{d_V}{|G|}
$$
where
$$
d_V = \sum_{\chi \in I(V)} \chi(e)^2.
$$
\end{thm} 

The result of the previous section tells us that parallel algorithms are optimal for symmetric oracle discrimination. 
\begin{thm}\label{mainSOD}
Suppose $G$ is a finite group and $\pi: G \to U(V)$ a unitary representation of $G$. Then an optimal $t$-query algorithm to solve symmetric oracle discrimination succeeds with probability
$$
P_{\rm opt} = \frac{d_{V^{\otimes t}}}{|G|}
$$
where
$$
d_{V^{\otimes t}} = \sum_{\chi \in I(V^{\otimes t})} \chi(e)^2.
$$
\end{thm}

\begin{proof} Theorem \ref{parallelqueries}  tells us that a $t$-query parallel algorithm achieves the optimal success probability. As noted this is equivalent to a single-query algorithm using the representation $\pi^{\otimes t}: G \to U(V^{\otimes t})$. Now apply Theorem \ref{singlequerySOD}. \end{proof} 

To express the exact and bounded error query complexity of symmetric oracle discrimination we're compelled to make the following definitions.\\

Let $V$ denote a $\C G$-module. The {\it quantum base size}, denoted $\gamma(V)$, is the minimum $t$ for which every irrep of $G$ appears in $V^{\otimes t}$. If no such $t$ exists then $\gamma(V) = \infty$. The {\it bounded error quantum base size}, denoted $\gamma^{\rm bdd}(V)$ is the minimum $t$ for which
$$
\frac{1}{|G|}\sum_{\chi \in I(V^{\otimes t})} \chi(e)^2 \geq 2/3.
$$ \\

If $(G, V, \pi)$ is a case of symmetric oracle discrimination then by Theorem \ref{mainSOD} the number of queries needed to produce a probability $1$ algorithm is $\gamma(V)$. That is, the exact quantum query complexity of the problem is equal to the quantum base size of $V$. Similarly the bounded error query complexity is $\gamma^{\rm bdd}(V)$. \\

It may happen that one of these quantities is infinite. However when $V$ is a faithful representation then a classical result attributed to Brauer and Burnside (\cite{isaacs:char}, Theorem 4.3) guarantees that every irrep of $G$ appears in one of the tensor powers $V^{\otimes 0}, V, V^{\otimes 2}, \dots, V^{\otimes m -1}$ where $m$ is the number of distinct values of the character of $V$. If $V$ contains a copy of the trivial representation, then we can say that every irrep of $G$ is contained in some tensor power $V^{\otimes t}$ for some $t$. Hence in this case (with $V$ faithful and containing a copy of the trivial irrep) both $\gamma(V)$ and $\gamma^{\rm bdd}(V)$ are finite. \\

In particular, this occurs whenever we ``quantize'' a classical symmetric oracle discrimination problem. This is the learning problem specified by a finite set $\Omega$ and a homomorphism $G \to \Sym(\Omega)$. A query to an oracle hiding $a \in G$ consists of inputting $\omega \in \Omega$ and receiving $a \cdot \omega$. The learner must determine the hidden group element (or permutation) $a$. The quantized learning problem uses the homomorphism $G \to U(\C \Omega)$ sending elements of $G$ to permutation matrices. (Such a representation is called a {\it permutation representation.}) Then the quantized learning problem is faithful if the original problem is faithful and the $\C G$-module contains a copy of the trivial representation, namely $\Span \{\sum_{\omega \in \Omega} |\omega \rangle\}$. \\

This is precisely the situation we would like to study because we can compare the classical and quantum query complexity. Classically the exact and bounded error query complexities are equal, since if a classical algorithm does not use enough queries to identify the hidden permutation with certainty then it must make a guess between at least $2$ equally likely permutations which behave the same on all the queries that were used, resulting in a success rate of at most $1/2$. \\

\begin{itemize} \item Suppose $\Omega = \{1, \dots, n\}$ hosts the defining permutation representation of $G = S_n$. Then $n - 1$ queries are required to determine a hidden permutation $\sigma$. \\

\item If we take the same action but restrict the group to $A_n \leq S_n$ then we need $n-2$ queries to determine a hidden element $\sigma \in A_n$. \\

\item Consider the action of the dihedral group $D_n$ on the set of vertices of an $n$-gon. Then $2$ queries are required to determine a hidden group element. \\
\end{itemize}

In general the classical query complexity is a well-known invariant of a permutation group $G$ denoted $b(G)$ called the {\it minimal base size} or just {\it base size} of $G$ \cite{LS:primbase}. It may be defined to be the length of the smallest tuple $(\omega_1, \dots, \omega_t) \in \Omega^t$ with the property that $(g \cdot \omega_1, \dots, g \cdot \omega_t) = (\omega_1, \dots, \omega_t)$ if and only if $g = 1$. From the definition it is clear that the base size agrees with the non-adaptive classical query complexity of the problem. In fact, it is also equal to the adaptive query complexity, since if a sequence of adaptive guesses $(\omega_1, \dots, \omega_t)$ suffices to identify a particular hidden $g \in G$, then the same sequence of guesses works for {\it every} element of the group. This means any optimal algorithm may be implemented non-adaptively. Thus the classical query complexity of symmetric oracle discrimination of $G \leq \Sym(\Omega)$ is the base size of $G$ and the quantum exact (bounded error) query complexity is the (bounded error) quantum base size. We are naturally led to a broad group theoretic problem: \\

{\it Question.} What are the relationships between $b(G), \gamma(\C \Omega)$ and $\gamma^{\rm bdd}(\C \Omega)$? \\

We are not aware of any direct comparison of these quantities in the group theory literature.
Here we only compute the various quantities for some special cases. We saw earlier that $b(S_n) = n -1$. We will prove

\begin{thm}\label{guesstheperm}  Let $\gamma, \gamma^{\rm bdd}$ denote the quantum base sizes for $S_n$ acting on $\{1, \dots, n\}$. Then
\begin{enumerate}
\item $\gamma = n - 1$ queries are necessary and sufficient for exact learning.
\item $\gamma^{\rm bdd} = n - 2\sqrt{n} +  \Theta(n^{1/6})$  queries are necessary and sufficient to succeed with probability $2/3$.
\item In fact,  for any $\epsilon \in (0,1)$, $n - 2\sqrt{n} +  \Theta(n^{1/6})$ queries are necessary and sufficient to succeed with probability $1-\epsilon$.

\end{enumerate}
\end{thm}
\begin{proof}
Recall that the irreducible characters of $S_n$ are parametrized by partitions of $n$ which can be written either as a sequnce $[\lambda_1, \dots, \lambda_n]$ or as a Young diagram with $n$ total boxes and $\lambda_i$ boxes in the $i$th row. Let $V = \C \{1, \dots, n\}$ denote the $\C G$-module  corresponding to the defining permutation representation of $S_n$. Then $V$ decomposes as a sum of two irreducibles:
$$
V = V_{[n]} \oplus V_{[n-1, 1]}.
$$
We note that $V_{[n]}$ is the trivial representation. A well-known rule says that if $V_{\lambda}$ is a simple representation corresponding to the Young diagram $\lambda$ then the irreps appearing in $V \otimes V_{\lambda}$ 
$$
I(V \otimes V_{\lambda}) = \{V_{\mu}\ |\ \mu \in \lambda^{\pm}\}.
$$
  where $\lambda^{\pm}$ is the set of Young diagrams obtained from $\lambda$ by adding then removing a box from lambda. In particular, this shows by induction that
$$
I(V^{\otimes t}) = \{V_{\mu}\ |\ \mu \text{ has at least $n-t$ columns} \}.
$$
We see that $n-1$ queries are required until every irreducible is contained in $V^{\otimes t}$ (in particular, the sign representation corresponding to the partition $[1^n] = [1,1,\dots,1]$ is not included in $V^{\otimes t}$ unless $t \geq n-1$). This proves part (1) of the theorem.

To prove part (2) we must examine more closely the set $I_t = I(V^{\otimes t})$ consisting of all partitions with at least $n-i$ columns (i.e. $\lambda_1 \geq n-i)$. We are interested in the sum
$$
d_t := d_{V^{\otimes t}} = \sum_{\chi \in I(V^{\otimes t})} \chi(e)^2.
$$
It is well known that if $\chi$ is an irrep corresponding to the Young diagram $\lambda$ then $\chi(e)$ is equal to the number of standard tableaux of shape $\lambda$ (\cite{sagan:symgroup},  Theorem 2.5.2). Hence $\chi(e)^2$ is equal to the number of pairs of standard tableaux of shape $\lambda$. Now by the Robinson-Schensted correspondence, the sum above is equal to the number of sequences of the numbers $\{1, \dots, n\}$ whose longest increasing subsequence is at least $n-t$ (see e.g. \cite{sagan:symgroup}, Theorem 3.3.2). Next, a deep result of Baik, Deift and Johannson \cite{BDJ:1999} identifies the distribution of the $l_n$, the length of the longest increasing subsequence of a random permutation of $n$ elements, as the Tracy-Widom distribution (which also governs the largest eigenvalue of a random Hermitian matrix) of mean $2 \sqrt{n}$ and standard deviation $n^{1/6}$. In particular, Theorem 1.1 of \cite{BDJ:1999} asserts that if $F(x)$ is the cumulative distribution function for the Tracy-Widom distribution, then 
$$
\lim_{n\to\infty}Prob\biggl(\frac{l_n-2\sqrt{n}}{n^{1/6}}\leq x\biggr) = F(x)
$$
Let  $c$ be any real number.  If we use $t=n - 2\sqrt{n} +  cn^{1/6}$ queries, then our success probability will be
$$
Prob(l_n\geq n-t) = 1 - Prob( l_n < 2\sqrt{n} -cn^{1/6}) = 1 - Prob\biggl(\frac{l_n-2\sqrt{n}}{n^{1/6}} < -c\biggr)  \to 1 - F(-c)
$$
Thus for any $\epsilon\in (0,1)$, if we wish to succeed with probability $1-\epsilon$, it will be necessary and sufficient to 
use $t=n - 2\sqrt{n} +  cn^{1/6}$ queries, where $c=-F^{-1}(\epsilon)$ (for $n$ sufficiently large).
\end{proof}

Here is the analogous result for identifying an element of the alternating group.  

\begin{thm}\label{guesstheeggplantperm}  Consider the standard action of $A_n$ acting on $\{1, \dots, n\}$. Then the quantum base sizes are given as follows.
\begin{enumerate}
\item $\gamma = n - \lceil \sqrt{n} \rceil$ are necessary for exact learning.
\item $\gamma^{\rm bdd} = n - 2\sqrt{n} +  \Theta(n^{1/6})$  are necessary and sufficient to succeed with probability $2/3$. In fact,  for any $\epsilon \in (0,1)$, $n - 2\sqrt{n} +  \Theta(n^{1/6})$ are necessary and sufficient to succeed with probability $1-\epsilon$.
\end{enumerate}
\end{thm}

\begin{proof}

Recall the following facts about the representation theory of $A_n$. The conjugate of a partition $\lambda$ is the partition $\lambda^*$ obtained by swapping the rows and columns of $\lambda$; in other words $\lambda^* = (\lambda^*_1, \lambda^*_2, \dots)$ where $\lambda_i^* = $ the number of boxes in the $i$th column of $\lambda$. For each partition $\lambda$ of $n$ that is not self-conjugate, i.e. $\lambda\neq\lambda^*$, the restriction of $V_{\lambda}$ to $A_n$ is an irreducible representation $W_{\lambda}$ of $A_n$.  Also, $W_{\lambda}= W_{\lambda^*}$.  For self conjugate $\lambda$, the representation $V^{\lambda}$ breaks up into two distinct irreducible representations $W_{\lambda}^+$ and $W_{\lambda}^-$ of equal dimension.

Recall from the previous proof that after $t$ queries, we get copies of all the $V_{\lambda}$ such that $\lambda_1\geq n-t$.  Observe that for any partition $\lambda$, we must have either $\lambda_1\geq \lceil \sqrt{n} \rceil$ or $\lambda^*_1\geq \lceil \sqrt{n} \rceil$.  (If both fail, the partition fits into a square of side length $\lceil \sqrt{n} \rceil - 1$, which contains fewer than $n$ boxes.) It follows that after $t= n - \lceil \sqrt{n} \rceil$ queries, for any $\lambda$, we have picked up a copy of $V_{\lambda}$ or $V_{\lambda^*}$.  Hence we have every irreducible representation of $A_n$. Therefore,  $ n - \lceil \sqrt{n} \rceil$ queries suffice for exact learning. Showing that that fewer queries cannot suffice is similar.  Here we make the observation that there exists a partition  $\lambda$ such that $\lambda_1 < \lceil \sqrt{n} \rceil + 1$ and $\lambda^*_1< \lceil \sqrt{n} \rceil + 1$, since $n$ boxes can be packed into a square of side length $\lceil \sqrt{n} \rceil $.  It follows that $t= n - \lceil \sqrt{n} \rceil - 1$ queries do not pick up the $V_{\lambda}$ or $V_{\lambda^*}$ for such $\lambda$.  Thus, we do not get every irrep of $A_n$. 

We now examine the bounded error case.  For a positive integer $t$, let $p_t$ be the success probability of the optimal $t$-query algorithm for identifying a permutation of $S_n$ and let $q_t$ be the corresponding probability for $A_n$.  

Let $V$ denote the $t$-fold tensor power of the defining representation of $S_n$. We can decompose $V$ as a direct sum of irreps of $S_n$ and if we know which $V_{\lambda}$ appear we can determine which irreps of $A_n$ appear in $V$.  In particular, each time we have a non-self-conjugate $\lambda$ such that $V_{\lambda}$ appears in $V$, we will have $W_{\lambda}$ appearing in $V$.   Let's consider the contribution of this appearance to the success probability $p_t$ and $q_t$, which is the square of the dimension divided by the order of the group.  Since the dimension of $V_{\lambda}$ equals the dimension of $W_{\lambda}$, while the order of $S_n$ is twice the order of $A_n$, the contribution to $q_t$ is twice the contribution to $p_t$. 
 
Now if $\lambda$ is self-conjugate then $V_{\lambda}$ decomposes into two irreps of $S_n$  of equal dimensions.  The sum of the squares of these two irreps is thus one-half the square of the dimension of $V_{\lambda}$.  Once we've divided by the sizes of the groups, we see that the contribution to $q_t$ is equal to the contribution to $p_t$.  
 
We have thus seen that for any $\lambda$ the contribution to $q_t$ is either 2 or 1 times the contribution to $p_t$.  It follows that 
$$
p_t \leq q_t \leq 2 p_t
$$
Thus for $q_t\geq 2/3$ we must have $p_t \geq 1/3$, which as we showed in Theorem \ref{guesstheperm} requires $n - 2\sqrt{n} +  \Theta(n^{1/6})$ queries.  On the other hand, if we are given  $n - 2\sqrt{n} +  \Theta(n^{1/6})$ queries, we achieve $p_t \geq 2/3$, which forces $q_t \geq 2/3$.
\end{proof}

The two theorems above show that there is very little speedup possible when trying to identify a permutation from the symmetric group or the alternating group. For the alternating group, one can at least get by with $\sqrt{n}$ fewer queries for exact quantum learning. Here there is an analogy to Van Dam's problem of exactly learning the value of an $n$-long bitstring using queries to its bits \cite{vanDam:1998}.  Exact learning requires $n$ queries. However, if we are guaranteed in advance that the parity of the string is even, then only $\lfloor n/2 \rfloor$ queries are required for exact learning.  To see this using the techniques of the current paper, we argue as follows.  Let $G$ be the subgroup $\Z_2^n$ consisting of all strings of even parity.  If we are allowed $t$ queries, then we can access those representations $\rho_x$ of $\Z_2^n$ corresponding to strings $x$ of Hamming weight less than or equal to $t$ (see also the remarks in Section 7.3).  If $\bar{x}$ is the bitwise complement of $x$, then $\rho_x$ and $\rho_{\bar{x}}$ take the same values on $G$.  Now, for any string $x$, one of $x$ and $\bar{x}$ will have Hamming weight less than or equal to $\lfloor n/2 \rfloor$.  Hence every representation of $G$ can be accessed by $\lfloor n/2 \rfloor$ queries to the oracle, and we will succeed with probability 1.

\section{Query complexity of coset identification}
\label{sec:5}
%
%
%

In this section we derive a formula for the optimal success probability of a $t$-query algorithm to solve coset identification. In light of our previous result on parallelizability (Corollary \ref{tQueryCor}), this boils down to finding a formula in the single-query case. This will directly generalize the single-query results of \cite{BCMP:sod} used in Section 4. \\

To state the result we fix some notation. Suppose $(G, V, \pi, f)$ is an instance of coset identification with $H$ the preimage of $f(e)$. Given an $H$-representation $W$ let $W ^{\uparrow}$ denote the induced representation of $W$  (which is a representation of $G$; see Section \ref{InducedReps} below for more details.) Likewise if $W$ is a $\C G$-module then we denote by $W^{\downarrow}$ the $\C H$-module obtained by restriction to $H$. Recall that if $V$ is a $\C G$-module then $I(V)$ denotes the set of all irreducible characters of $G$ appearing in $V$. We sometimes use the notation $I_G(V), I_H(V)$ to emphasize which group we are considering. Finally, given two representations $A$ and $B$ we let
$$
A_B := \text{ the maximal subrepresentation of $A$ such that } I(A_B) \subseteq I(B).
$$
Thus $A_B$ denotes the sum of all the isotypical components of $A$ which correspond to an irreducible isotype appearing in $B$. We will be interested in the quantities
$$
\frac{\dim A_B}{\dim A}
$$\
which can be understood as the fraction of $A$ which is shared with $B$. \\

\begin{thm}\label{CIThm}
An optimal single-query algorithm to solve the instance $(G, V, \pi, f)$ of coset identification succeeds with probability
$$
P_{\rm opt} = \max_{Y \in \Irr(H)} \frac{\dim\ (Y^{\uparrow})_V}{\dim\ Y^{\uparrow}}.
$$
\end{thm}

In words: to find the optimal success probability, you look at an irrep $Y$ of $H$ which appears in $V^{\downarrow}$. Then you examine the fraction of $Y^{\uparrow}$ which is shared with $V$. Finally take the maximum over all irreps $Y$ appearing in $V^{\downarrow}$. \\

From this theorem we can quickly deduce Theorem \ref{mainSOD}, the single-query result for symmetric oracle identification. This is the special case when $H$ is the trivial group. Then $H$ has only one irrep, namely the trivial representation $\mathbbm{1}$, and $\mathbbm{1}^{\uparrow}$ is isomorphic to $\C G$. Hence the formula we get from Theorem \ref{CIThm} is
$$
P_{\rm opt} = \frac{\dim (\C G)_V}{|G|} = \frac{1}{|G|}\sum_{\chi \in I(V)
} \chi(e)^2
$$
which is the formula of Theorem \ref{mainSOD}. \\

The next two sections are devoted to the proof of Theorem $\ref{CIThm}$. First we prove the lower bound (i.e. existence of a state and measurement achieving the desired success probability) and then we prove the upper bound (optimality of that success probability).

\subsection{The lower bound}
First we collect some facts concerning induced representations and averaging operators needed for the proof of Theorem \ref{CIThm}. A fine treatment of the subject is contained in Serre's book \cite{serre:linear}. \\

\subsubsection{Induced Representations}\label{InducedReps}
Suppose $H$ is a subgroup of a finite group $G$ and let $Y$ denote a representation of $H$. Note that $\C G$ admits a right $H$-action. The representation of $G$ {\it induced from $Y$} is
$$
Y^{\uparrow_H^G} = \C G \otimes_{\C H} Y.
$$
When $H$ and $G$ are understood we simply write $Y^{\uparrow}$. Similarly if $W$ is a representation of $G$ then it is also a representation of $H$, called the restriction of $W$ to $H$. We denote it by $W^{\downarrow^G_H}$ or simply $W^{\downarrow}$. \\

From the definition of induced representations, we can write
$$
Y^{\uparrow} = \bigoplus_t t \otimes Y
$$
where $t$ ranges over a set of coset representatives for $H$. Conversely, if a representation $W$ of $G$ contains an $H$-invariant subspace $W_0$ such that
$$
W = \bigoplus_t tW_0
$$
where $t$ again ranges over a set of coset representatives for $H$, then $W$ is isomorphic to $W_0^{\uparrow}$ as $G$ representations. \\

In our situation all representations are unitary. In particular if $Y$ is a unitary representation of $H$ then $Y^{\uparrow}$ is equipped with the inner product determined by requiring the subspaces $t \otimes Y$ to be pairwise orthogonal, and translating the inner product of $Y$ to each subspace $t \otimes Y$. With this inner product $Y^{\uparrow}$ is a unitary representation of $G$. We will often denote the orthogonal projection onto $e \otimes Y$ by $E$. Then the orthogonal projection onto $t \otimes Y$ is $tEt^{-1}$, and we have $\sum_t tEt^{-1} = I$. \\
\subsubsection{Averaging operators}\label{sec:avgops}
Given a $\mathbb{C} G$-module $V$ we can define the averaging operator, which turns an arbitrary linear map $A: V \to V$ into a $G$-invariant one:
\begin{align*}
R_G&:\ \End_{\mathbb{C}}(V) \to \End_G(V) \\
R_G(A)& := \frac{1}{|G|}\sum_{g \in G}gAg^{-1}.
\end{align*}
Note that $R_G(A)$ commutes with every $g \in G$ so that indeed $R_G(A)$ is $G$-invariant, i.e. $R_G(A) \in \End_G(V)$. If $B$ is a $G$-invariant operator then $R_G(BA) = BR_G(A)$. The map $R_G$ is trace-preserving, so in particular if $p$ is a projection then $R_G(p)$ is non-zero, since it has non-zero trace. If $V$ contains only a single isotype of irrep, i.e. $V \cong Y \otimes \mathbb{C}^m$ for some irrep $Y$ then $R_G$ is closely related to the partial trace with respect to the subspace $Y$:
\begin{equation}\label{eq-averagepart}
R_G(A) = \frac{1}{\dim Y}I \otimes \Tr_{Y}(A).
\end{equation}\\
\subsubsection{Proof of the lower bound}

Before giving the proof of the lower bound in Theorem 5.1 we prove a preliminary proposition. \\

If $R$ is a algebra over $\C$, $V$ an $R$-module and $W \leq V$ a linear subspace, we let $R \cdot W$ denote the submodule of $V$ generated by $W$ (i.e. the smallest submodule containing the subspace $W$). Similarly for $r \in R$ we let $r \cdot W$ denote the subspace $\{rw:\ w \in W\}$.

\begin{prop}\label{inducedProp1}
Suppose $Y$ is an irreducible unitary representation of $H$ (a subgroup of $G$). Also suppose $V$ is a $G$-subrepresentation of $Y^{\uparrow}$. Let $E$ denote orthogonal projection onto $e \otimes Y \subset Y^{\uparrow}$. Then there exists a unit vector $\psi \in V$ such that
$$
\langle \psi | E | \psi \rangle = \frac{\dim V}{\dim Y^{\uparrow}}.
$$
\end{prop}
{\bf Remark.} In Proposition \ref{inducedProp2} we will prove this is an upper bound for $\langle \psi | E | \psi \rangle$ over all unit vectors $\psi \in V$. \\
\begin{proof} Let $\Pi_V$ denote the $G$-invariant orthogonal projection onto $V$. Since $Y$ is irreducible, $E$ is a minimal idempotent in $\End_H(Y^{\uparrow})$. Therefore, since $\Pi_V$ also belongs to $\End_H(Y^{\uparrow})$, we know $E \Pi_V E$ is a scalar times $E$. In turn this implies $\Pi_VE\Pi_V$ is a scalar multiple of an orthogonal projection, since it is self-adjoint and
$$
(\Pi_V E \Pi_V)^2 = \Pi_V E \Pi_V E \Pi_V = \text{ scalar} \cdot \Pi_V E \Pi_V.
$$

The image of $\Pi_V E \Pi_V$ is an $H$-invariant subspace of $V$ which is either $0$ or isomorphic to $Y$. Let this subspace be $Y'$, so we have
\begin{equation}\label{eq:minidem}
\Pi_V E \Pi_V = \lambda \Pi_{Y'}
\end{equation}
for some non-zero scalar $\lambda \in \C$. 
We will also use the fact that
\begin{equation}\label{eq:prop52}
\Pi_{Y'}E\Pi_{Y'} = \lambda \Pi_{Y'},
\end{equation}
which results from Eq. (\ref{eq:minidem}) by multiplying the equation by $\Pi_{Y'}$ on the left and right. Next, we claim that $Y'$ is not zero (so it is in fact isomorphic to $Y$ as an $H$-module). Indeed, we have
$$
\C G \cdot Y' = \sum_t t \cdot Y' = \sum_t \im(\Pi_V t E t^{-1} \Pi_V) \supset \im(\Pi_V (\sum_t tEt^{-1} )\Pi_V) = \im(\Pi_V) = V
$$
where the sum is over a set of coset representatives of $H$. This shows that $Y'$ is non-zero. In particular we have $\dim Y' = \dim Y$. We can now compute the scalar $\lambda$ via

\begin{align*}
\dim V &= \Tr(\Pi_V) = \Tr(\sum_t \Pi_V tEt^{-1} \Pi_V) \quad \quad \text{(since $\sum_t t E t^{-1} = I)$} \\
&= \Tr(\sum_t t \Pi_V E \Pi_V t^{-1}) \quad \quad \text{(since $\Pi_V$ commutes with the action of $G$)} \\
&= \lambda \sum_t \Tr(t \Pi_{Y'} t^{-1}) \quad \quad \text{(by Eq. \ref{eq:minidem})} \\
&= \lambda |G:H |\dim Y' = \lambda \dim Y^{\uparrow}
\end{align*}

which yields $\lambda = \frac{\dim V}{\dim Y^{\uparrow}}$. \\

Finally, let $|\psi \rangle$ be any unit vector in $Y'$. Consider the rank-1 projection $|\psi \rangle \langle \psi|: Y' \to Y'$. We apply the averaging operator $R_H$ (see Section \ref{sec:avgops}) to get $R_H(|\psi \rangle \langle \psi|) = \frac{1}{|H|}\sum_{h \in H} h |\psi \rangle \langle \psi| h^{-1}$. The space of $H$-invariant maps from $Y'$ to $Y'$ is 1-dimensional (by Schur's Lemma) and spanned by $\Pi_{Y'}$. Hence $R_H(|\psi \rangle \langle \psi|)$ is a scalar multiple of $\Pi_{Y'}$, and by taking traces we find $R_H(|\psi \rangle \langle \psi|) = \frac{1}{\dim Y}\Pi_{Y'}$. 

Using this we compute
\begin{align*}
\langle \psi | E \psi \rangle &= \Tr(|\psi \rangle \langle \psi| E) = \Tr( \frac{1}{|H|}\sum_{h \in H} h|\psi \rangle \langle \psi| h^{-1} E ) \quad \quad \text{(since $E$ is $H$-invariant)} \\
&=  \Tr(\frac{1}{\dim Y}\Pi_{Y'} E) \quad \quad \quad \quad \text{(by above discussion)} \\
&= \frac{1}{\dim Y} \Tr(\Pi_{Y'} E \Pi_{Y'}) \\ 
&= \frac{1}{\dim Y} \Tr(\lambda \Pi_{Y'}) \quad \quad \quad \quad \text{(by Eq. \ref{eq:prop52} above)} \\
&= \frac{\dim V}{\dim Y^{\uparrow}} \quad \quad \quad \quad \text{(since $\Tr(\Pi_{Y'}) = \dim Y$, and $\lambda = \frac{\dim V}{\dim Y^{\uparrow}}$)}
\end{align*}

as needed. \end{proof}

\begin{proof}[Proof of Theorem \ref{CIThm}, lower bound] Let $Y$ be an irreducible constituent of $V_{\downarrow}$ which maximizes the quantity
$$
\frac{\dim (Y^{\uparrow})_V}{\dim Y^{\uparrow}}.
$$
Let $V'$ denote the $G$-subrepresentation $(Y^{\uparrow})_V$ of $Y^{\uparrow}$ and again let $E$ denote the orthogonal projection onto the subspace $e \otimes Y \subset Y^{\uparrow}$. Then by Proposition \ref{inducedProp1} there exists a unit vector $|\psi \rangle \in V'$ such that
$$
\langle \psi | E | \psi \rangle = \frac{\dim V'}{\dim Y^{\uparrow}} = \frac{\dim Y^{\uparrow}_V}{\dim Y^{\uparrow}}.
$$

Now consider the oracle problem given by $(G, V', \pi', f)$ (i.e. the coset identification problem where the oracle is represented on $V'$ rather than $V$). Let $\Pi_{V'}$ denote the $G$-invariant orthogonal projection onto $V'$. We define a single-query algorithm for $(G, V', \pi', f)$ using no ancilla, the input state $|\psi \rangle$, and projective measurement $\{t \Pi_{V'}E \Pi_{V'}t^{-1}\}_t$ where $t$ ranges over a set of coset representatives for $H$ (so measuring outcome $t$ uniquely determines a coset of $H$). The measurement is used to distinguish the density operators $\{\rho_t = t \rho t^{-1}\}$ where $\rho = \frac{1}{|H|}\sum_{h \in H} h |\psi \rangle \langle \psi | h^{-1}$. Note that the support of $\rho$ is contained in $V'$, since $|\psi \rangle \in V'$ and $V'$ is $G$-invariant. Therefore $\rho \Pi_{V'} = \Pi_{V'}\rho = \rho$. Using this, we compute the success probability as
\begin{align*}
P_{\rm succ} &= \frac{1}{|G:H|}\sum_t \Tr(\rho_t t \Pi_{V'} E t^{-1} \Pi_{V'}) \\
&= \frac{1}{|G:H|} \sum_t \Tr(\rho_t t \Pi_{V'}E \Pi_{V'} t^{-1}) \quad \quad \text{ (since $\Pi_{V'}$ is $G$-equivariant, so commutes with $t^{-1}$)}\\
&= \frac{1}{|G:H|} \sum_t \Tr(\rho \Pi_{V'} E \Pi_{V'}) \quad \quad \text{(since the trace is cyclic, and $t^{-1}\rho_t t= \rho)$} \\
&= \Tr(\rho E) \quad \quad \text{(since the trace is cyclic, and $\rho \Pi_{V'} = \Pi_{V'}\rho = \rho$)} \\
&= \langle \psi | E |\psi \rangle = \frac{\dim Y^{\uparrow}_V}{\dim Y^{\uparrow}}.
\end{align*}
This shows that there is an algorithm for $(G, V', \pi', f)$ which succeeds with probability $\frac{\dim Y^{\uparrow}_V}{\dim Y^{\uparrow}}$. Since $V' = (Y^{\uparrow})_V$ only contains irreps which are also contained in $V$, Lemma \ref{lem:monotone} implies there is also an algorithm for $(G, V, \pi, f)$ which succeeds with the same probability.
\end{proof}
{\bf Remark.} In applying Lemma \ref{lem:monotone} to produce an algorithm for $(G, V, \pi, f)$, one may have to introduce an ancilla register, to ensure that irreps appear with sufficiently large multiplicity to allow an embedding of $V'$ into the workspace.

\subsection{The upper bound} In this section we prove the upper bound of Theorem \ref{CIThm} using a minimum-error quantum state discrimination approach \cite{EMV:2006}. Before explaining the strategy to obtain the bound, we review the set-up. We fix an instance of coset identification $(G, V, \pi, f:G \to X)$. The subgroup $H$ is the preimage of $f(e)$, and the elements of $X$ may be identified with the left cosets of $H$. A single-query algorithm uses an initial state $|\psi \rangle \in V$ and feeds it to the oracle, which is a hidden element $a$ sampled uniformly from $G$. Afterwards, a measurement $\{E_x\}_{x \in X}$ is applied with the goal of recovering $f(a)$. With a choice of initial state fixed, the task of finding an optimal measurement $\{E_x\}$ amounts to finding an optimal measurement to discriminate the mixed states $\{\rho_x\}_{x \in X}$, where
\begin{equation*}
    \rho_x = \frac{|X|}{|G|}\sum_{\substack{g \in G \\ f(g) = x}}g|\psi \rangle \langle \psi |g^{-1}.
\end{equation*}
Indeed, the success probability of the algorithm is equal to the probability that the measurement $\{E_x\}$ successfully discriminates the mixed states $\{\rho_x\}$, namely
$$
P_{\rm succ} = \frac{1}{|X|}\sum_{x \in X}\Tr(E_x \rho_x).
$$
We will prove that this success probability is bounded above by the quantity given in Theorem \ref{CIThm}, which involves induced representations. We now provide an outline of the proof to give an indication of how induced representations enter the picture. \\

To take advantage of symmetry in the problem, note that the density matrices $\{\rho_x\}$ always satisfy
$$
\rho_{g \cdot x} = g \rho_x g^{-1}.
$$
We say a set of operators with this symmetry is {
\it orbital} (a precise definition is given below). We first argue that any optimal measurement to distinguish an orbital set of density matrices can be modified to produce another optimal measurement which is itself orbital (Lemma \ref{lem:orblem}). Next we aim to simplify the problem further by showing that any orbital POVM can be replaced by a measurement which is both orbital {\it and} projective. To do so requires embedding the original $\C G$-module $V$ into a larger one $W$ by adding an ancilla register. This is the content of Lemma \ref{equivMsmt}, which is a ``symmetric" version of the usual result that any POVM can be simulated using projective measurements and ancilla registers. As a result of this lemma we may make the following assumptions about an optimal single query algorithm, which uses the larger Hilbert space $W$:
\begin{enumerate}
    \item The measurement operators $\{E_x\}_{x \in X}$ are projective and orbital.
    \item The initial state $|\psi \rangle$ belongs to a $G$-invariant subspace of $W$ which is isomorphic to $V$.
\end{enumerate}
The existence of a projective orbital measurement implied by (1) is a strong condition on the structure of $W$: using the completeness relation, $W$ can be written as the direct sum of the images of the measurement operators $\{E_x\}$. The subspace $Y$ which is the image of $E_{x_0}$ is left invariant by $H$, and the other subspaces are obtained through translation by a coset representative. This realizes $W$ as the induced representation $Y^{\uparrow}$. Finally, in this restricted setting (incorporating the assumption (2) that $|\psi \rangle \in V$) we are able to bound the success probability by decomposing $Y$ into irreducible $H$-subrepresentations, and then applying a critical inequality (Proposition \ref{inducedProp2}) that covers the situation when $Y$ is irreducible. We now give the details. \\

With a given unitary representation $V$ of $G$ and a fixed $G$-set $X$ understood we say a set of operators $\{A_x\}_{x \in X}$ (on $V$) is {\it orbital} if $gA_{x}g^{-1} = A_{g \cdot x}$ for all $x$ and $g$. The density matrices for a single query algorithm for coset identification form an orbital set. \\

\begin{lem}\label{lem:orblem}
Suppose $\{\rho_x \}_{x \in X}$ is an orbital set of density matrices. Then there exists an optimal measurement to distinguish the states $\{\rho_x\}$ which is orbital.
\end{lem}

\begin{proof} Eldar, Megretski, Verghese  give the proof when $X = G$ with the action of left multiplication (\cite{EMV:2004}, Section 4.3) and it works in this setting as well. We give the proof for the reader's convenience. Suppose $\{E_x\}_{x \in X}$ is an optimal measurement. Then we define new measurement operators $\{\widehat{E}_x\}_x$ by
$$
\widehat{E}_x = \frac{1}{|G|}\sum_{g \in G}gE_{g^{-1} \cdot x}g^{-1}.
$$
We claim that $\{\widehat{E}_x\}$ is an orbital POVM which discriminates the states $\{\rho_x\}$ with the same success probability as $\{E_x\}$. Each operator $\widehat{E}_x$ is a nonnegative combination of positive semi-definite operators, hence is positive semi-definite. They satisfy the completeness relation:
\begin{equation*}
    \sum_{x \in X} \widehat{E}_x = \frac{1}{|G|}\sum_{\substack{x \in X \\ g \in G}}g E_{g^{-1} \cdot x}g^{-1} = \frac{1}{|G|}\sum_{g \in G}I = I.
\end{equation*}
The completness relation for $\{E_x\}$ is used in the second equality. We check that the POVM $\{\widehat{E}_x\}$ is orbital:
\begin{equation*}
    h\widehat{E}_x h^{-1} = \frac{1}{|G|}\sum_{g \in G} hgE_{g^{-1} \cdot x}g^{-1}h^{-1} = \frac{1}{|G|}\sum_{k \in G}kE_{k^{-1}h \cdot x}k^{-1} = \widehat{E}_{h \cdot x}.
\end{equation*}
To complete the proof it suffices to show that the new measurement discriminates the states $\{\rho_x\}$ with the same probability as the original measurement. Indeed, we have
\begin{align*}
\frac{1}{|X|}\sum_{x \in X}\Tr(\widehat{E}_x \rho_x) &= \frac{1}{|X||G|}\sum_{x \in X}\sum_{g \in G}\Tr(gE_{g^{-1} \cdot x}g^{-1} \rho_x) = \frac{1}{|X||G|}\sum_{g \in G}\sum_{x \in X}\Tr(E_{g^{-1} \cdot x}\rho_{g^{-1} \cdot x}) = \\
&= \frac{1}{|G|} \sum_{g \in G} \left(\frac{1}{|X|} \sum_{y \in X} \Tr(E_y \rho_y)\right) = \frac{1}{|X|}\sum_{y \in X}\Tr(E_y \rho_y).
\end{align*}
The second equality follows from the orbital assumption $\rho_{g^{-1} \cdot x} =g^{-1} \rho_x g$, and the other steps are index substitutions.
\end{proof} 

For the next result we use the following fact (cf. \cite{NC:2011}, Exercise 2.67):
\begin{lem}
Suppose $V$ is a unitary representation of $G$, $W$ a subrepresentation, and $C: W \to V$ a $\C G$-module map which preserves inner products. Then $C$ can be extended to $V$, meaning there is a unitary $\C G$-module isomorphism $U: V \to V$ such that $U$ coincides with $C$ on $W$.
\end{lem}
\begin{proof}
Let $Y$ denote the orthogonal complement of $W$ and $Y'$ the orthogonal complement of $C(W)$. Since $C$ preserves inner products, it is injective, so $C(W) \cong W$ as $\C G$-modules. Hence $Y \cong Y'$ as $\C G$-modules, and there exists an inner product preserving isomorphism $D: Y \to Y'$. Now the desired unitary operator $U$ is given by
$$
U(x) = \begin{cases} C(x) & x \in W \\
D(x) & x \in Y.
\end{cases}
$$
\end{proof}

The following result is an equivariant version of the argument given by Chuang and Nielsen to show that arbitrary measurement operators can be simulated using projective measurements and ancilla spaces (see \cite{NC:2011}, Section 2.2.8). 
\begin{lem}\label{equivMsmt}
Suppose $\{E_x\}$ is an orbital POVM on the space $V$. Then there exists a unitary representation $W$ and an inner product preserving $\C G$-module embedding
$$
\iota: V \to W
$$
together with a projective orbital measurement $\{\overline{E}_x\}$ on $W$ such that for any state $|\psi \rangle$, the measurement statistics by measuring $|\psi \rangle$ with $\{E_x\}$ are identical to those given by the state $\iota |\psi \rangle$ and measurement $\{\overline{E}_x\}$.
\end{lem}

\begin{proof} Let $W$ be the space $V \otimes \C G$ and fix a basepoint $x_0$ of the $G$-set $X$.Let $M_x = (E_x)^{1/2}$ be the non-negative square root of $E_x$. The uniqueness of square roots implies that the set $\mathcal{M} = \{M_x\}$ is orbital. In addition, these constitute a set of measurement operators for the POVM, meaning $M_x^*M_x = E_x$. \\

Now let $C_{\M}$ be the {\it controlled-$\M$ operator} acting on $W$ via
$$
C_{\M}|\psi, g \rangle = \sqrt{|X|} M_{g \cdot x_0}|\psi \rangle \otimes |g \rangle.
$$
  Note that $C_{\M}$ is a $\C G$-module endomorphism of $W = V \otimes \C G$, since
\begin{align*}
C_{\M}(h \cdot |\psi, g \rangle) = C_{\M}|h \psi, hg \rangle &= \sqrt{|X|}M_{hg \cdot x_0}h|\psi \rangle \otimes |hg \rangle \\ 
&= \sqrt{|X|}hgM_{x_0}g^{-1}|\psi \rangle \otimes |hg \rangle = h \cdot C_{\M}|\psi, g \rangle.
\end{align*}
For the third equality we used the fact that $\M$ is orbital, i.e. $M_{g \cdot x_0} = gM_{x_0}g^{-1}$. Now $C_{\M}$ is not necessarily invertible, but we claim that $C_{\mathcal{M}}$ preserves inner products on the subspace $V \otimes |\eta \rangle$, where $|\eta \rangle = \frac{1}{\sqrt{|G|}}\sum_{g \in G}|g \rangle$ is the equal superposition vector in $\C G$:
\begin{align*}
    \langle C_{\M}\left(|\psi , \eta \rangle \right)|\ C_{\M}\ |\phi, \eta \rangle
    &= \frac{1}{|G|}\sum_{g,h \in G} \langle C_{\M}\left(|\psi, g \rangle\right) |\ C_{\M}\ |\phi, h \rangle \quad \quad \text{(by the def. of $|\eta \rangle$)} \\ &= \frac{|X|}{|G|}\sum_{g,h \in G}\langle M_{g \cdot x_0}|\psi\rangle \otimes | g \rangle\ |\ M_{h \cdot x_0}| \phi\rangle \otimes |h \rangle \rangle \quad \quad \text{(by the def. of $C_{\M}$)} \\
    &= \frac{|X|}{|G|}\sum_{g \in G} \langle \psi | E_{g \cdot x_0} | \phi \rangle \quad \quad \text{(since $\langle g | h \rangle = \delta_{gh}$ and $M_x^*M_x= E_{x}$)} \\ &= \frac{|X|}{|G|}\sum_{h \in H}\sum_{x \in X} \langle \psi | E_x |\phi \rangle \quad \quad \text{(by writing $g = th$ where $t \cdot x_0 = x)$} \\
    &= \frac{|X|}{|G|}\sum_{h \in H} \langle \psi | \phi \rangle = \langle \psi | \phi \rangle \quad \quad \text{(by the completeness relation and $|H| = |G|/|X|$).}
\end{align*}
Therefore, by the previous lemma, there exists a unitary $\C G$-module endomorphism $U$ which restricts to $C_{\M}$ on $V \otimes |\eta \rangle$. We are ready to define the embedding $\iota$ and measurement $\{\overline{E}_x\}$ that satisfy the claim of the theorem.
\\

We take $\iota$ to be the inclusion of $V$ as $V \otimes |\eta \rangle$:
$$
\iota |\psi \rangle = |\psi \rangle \otimes | \eta \rangle.
$$
Clearly $\iota$ is an inner product preserving $\C G$-module embedding. We define the projective measurement $\{\overline{E}_x\}$ by
$$
\overline{E}_x = U^{-1} \left(\sum_{g: g \cdot x_0 = x} I \otimes |g \rangle \langle g| \right)U.
$$
Here $I$ denotes the identity on $V$. The operators $\{\overline{E}_x\}$ constitute a projective measurement, and we check that they form an orbital set. Let $h \in G$. Then
\begin{align*}
h\overline{E}_{x}h^{-1} &= U^{-1}h\left(\sum_{g: g \cdot x_0  = x}I \otimes |g \rangle \langle g| \right)h^{-1}U \quad \quad \text{(since $U$ is a $\C G$-module map)} \\
&= U^{-1}\left(\sum_{g: g \cdot x_0 = x} I \otimes |hg \rangle \langle hg| \right)U \quad \quad \text{(since $h$ acts by $\pi_V(h) \otimes h$ on $W = V \otimes \C G$)} \\
&= U^{-1}\left(\sum_{k: k \cdot x_0 = h}I \otimes 
|k \rangle \langle k | \right)U
=\overline{E}_{h \cdot x}.
\end{align*}

Now suppose $\iota |\psi \rangle = |\psi \rangle \otimes |\eta \rangle$ is measured with the projective measurement $\{\overline{E}_x\}$. Then the probability of reading outcome $x$ is
\begin{align*}
\langle \psi, \eta &| \overline{E}_x | \psi, \eta \rangle = \langle \psi, \eta | U^{-1} \left(\sum_{g: g \cdot x_0 = x}I \otimes |g \rangle \langle g| \right)U|\psi, \eta \rangle \\
&= \big\langle C_{\M} (|\psi, \eta \rangle) | \left(\sum_{g: g \cdot x_0 = x}I \otimes |g \rangle \langle g|\right) C_{\M}|\psi, \eta \big \rangle \\
&= \frac{|X|}{|G|}\big\langle \left(\sum_{h \in G}M_{h \cdot x_0}|\psi \rangle \otimes |h \rangle \right) |\left(\sum_{g: g \cdot x_0 = x}I \otimes |g \rangle \langle g| \right) \left(\sum_{h' \in G}M_{h' \cdot x_0}|\psi \rangle \otimes |h' \rangle \right) \big\rangle \\
&= \frac{|X|}{|G|}\sum_{g: g \cdot x_0 = x}\langle M_{g \cdot x_0} \psi| M_{g \cdot x_0} \psi \rangle = \frac{|X|}{|G|}\sum_{g: g \cdot x_0 = x}\langle M_x \psi | M_x \psi \rangle \\
&= \langle \psi | E_x |\psi \rangle.
\end{align*}
The first three equalities are definitions, the fourth expands the multiplication, the fifth is notational and the last follows since the number of $g$ for which $g \cdot x_0 = x$ is equal to $|G|/|X|$ for all $x \in X$ (since $X$ is a transitive $G$-set). This proves the lemma. \end{proof}

As a result of the lemma, any orbital measurement to distinguish orbital states in a $\C G$-module $Y$ can be simulated by a projective orbital measurement in a larger $\C G$-module $W$. The next lemma explains that the existence of a projective orbital measurement implies a decomposition of the Hilbert space $W$ that realizes $W$ as a representation induced from $H$.

\begin{lem}\label{MsmtInduced} Suppose $\{E_x\}_{x \in X}$ is a projective orbital measurement on a $\C G$-module $W$. Let $W_x$ denote the image of $E_x$. Then $W_{f(e)}$ is an $H$-representation and $W \cong W_{f(e)}^{\uparrow}$.
\end{lem}

\begin{proof} If $\{E_x\}$ is an orbital measurement then $hE_{f(e)}h^{-1} = E_{h \cdot f(e)} = E_{f(e)}$ for all $h \in H$, i.e. $E_{f(e)}$ is a $\C H$-module homomorphism. Hence the image of $E_{f(e)}$ is invariant under $H$. \\

Since the set $\{E_x\}_x$ constitutes a measurement, $W = \bigoplus_x W_x$. Furthermore, since $E_{g \cdot f(e)} = gE_{f(e)}g^{-1}$, we have $W_{g \cdot f(e)} = gW_{f(e)}$. Hence $W = \bigoplus_t tW_{f(e)}$ where the sum is over a set of left coset representatives for $H$. By the characterization of induced representations discussed in Section 5, this shows $W \cong W_{x_0}^{\uparrow}$. \end{proof} 

The lemmas above show that as long as we are willing to embed our original representation $V$ into a larger representation $W$, we may assume that $W$ is induced from some representation $Y$ of $H$ and that the measurement operators are projections corresponding to the direct sum decomposition of $W$ as an induced representation. In other words, the measurement operator corresponding to outcome $x \in X$ is projection onto $t \otimes Y$ where $t$ is any element such that $t \cdot f(e) = x$. The next lemma is the final key to unlocking the upper bound.

\begin{prop}\label{inducedProp2}
Suppose $Y$ is an irreducible unitary representation of $H$ (a subgroup of $G$). Let $V$ be a $G$-subrepresentation of $Y^{\uparrow}$. Let $E$ denote orthogonal projection onto the subspace $e \otimes Y \subset Y^{\uparrow}$. Then for any unit vector $\psi \in V$ we have
$$
\langle \psi | E | \psi \rangle \leq \frac{\dim V}{\dim Y^{\uparrow}}.
$$
\end{prop}

\begin{proof} Consider the action of $H$ on $Y^{\uparrow}$ and let $W$ denote the $Y$-isotypic component of $Y^{\uparrow}$. Then $W \cong Y \otimes \mathbb{C}^m$ as $H$-representations where $m$ is the multiplicity of the irrep $Y$ in $Y^{\uparrow \downarrow}$. Since $E$ is an $H$-invariant projection with image isomorphic to $Y$, we may assume that $\psi$ belongs to $W$ in addition to $V$. (Indeed, the support of $E$ is contained in $W$, so $E = \Pi_W E = E \Pi_W$, which implies $\langle \psi | E | \psi \rangle = \langle \Pi_W \psi | E \Pi_W | \psi \rangle$.) Fix an orthonormal basis $\{y_1, \dots, y_d\}$ of $Y$ so that we may write
$$
|\psi \rangle = \sum_{i=1}^d\lambda_i |y_i, u_i \rangle
$$
where the $u_i$'s are unit vectors in $\mathbb{C}^m$ and $\lambda_i \geq 0$ with $\sum_i \lambda_i^2 = 1$. We apply the averaging operator $R_H$ of Section \ref{sec:avgops} to the projection $| \psi \rangle \langle \psi |$. By Equation (\ref{eq-averagepart}) of Section \ref{sec:avgops} we have
$$
R_H(|\psi \rangle \langle \psi |) = \frac{1}{\dim Y}\sum_{i=1}^d \lambda_i^2 \left(I \otimes |u_i \rangle \langle u_i|\right).
$$
In particular $R_H(|\psi \rangle \langle \psi|) \leq \frac{1}{\dim Y}\Pi_W$. Note that since $|\psi \rangle \in V \cap W$, the support of $R_H(|\psi \rangle \langle \psi |)$ is also contained in the $H$-submodule $V \cap W$. Hence we deduce the stronger inequality
$$
R_H(|\psi \rangle \langle \psi|) \leq \frac{1}{\dim Y}\Pi_{V \cap W}.
$$

Now we may estimate $\langle \psi | E \psi \rangle$:
\begin{align*}
\langle \psi | E \psi \rangle &= \Tr(E |\psi \rangle \langle \psi |) = \Tr(R_H(E |\psi \rangle \langle \psi |) ) \quad \quad \text{ (since $R_H$ preserves traces}) \\
&= \Tr(E R_H(|\psi \rangle \langle \psi | ) \quad \quad (\text{since $E$ is $H$-invariant}) \\
&\leq \frac{1}{\dim Y} \Tr(E \Pi_{V \cap W})\\
&\leq \frac{1}{\dim Y} \Tr(E \Pi_V).
\end{align*}
Here $\Tr(E \Pi_V)$ can be computed by averaging over a set of coset representatives for $H$:
\begin{align*}
\Tr(E \Pi_V) &= \frac{1}{|G:H|}\sum_t \Tr(t E \Pi_V t^{-1})
\end{align*}
Using that $\Pi_V$ commutes with the action of $G$ and that $\sum_t tEt^{-1} = I$ we have
\begin{align*}
\Tr(E \Pi_V) &= \frac{1}{|G:H|}\sum_t \Tr(t E t^{-1}\Pi_V) = \\
&= \frac{1}{|G:H|}\Tr(\Pi_V) = \frac{\dim V}{|G:H|}.
\end{align*}
Therefore
$$
\langle \psi | E \psi \rangle \leq \frac{\dim V}{\dim Y |G:H|} = \frac{\dim V}{\dim Y^{\uparrow}}.
$$
\end{proof}

We are ready to prove the upper bound in Theorem \ref{CIThm}.

\begin{proof}[Proof of Theorem \ref{CIThm}, upper bound] Let $(G, V, \pi, f)$ specify an instance of coset identification and let $H$ denote the stabilizer of a chosen point $x_0 \in X$ (recall that the codomain of $f$ is a transitive $G$-set $X$). Suppose an optimal single-query algorithm is given by an input state $|\psi \rangle \in V$ (again we may assume there is no workspace by absorbing it into $V$) and POVM $\{\widehat{E}_x\}$. By Lemmas \ref{equivMsmt} and \ref{MsmtInduced}, there is a (not necessarily irreducible) representation $Y$ of $H$ and $\C G$-submodule of $Y^{\uparrow}$ isomorphic to $V$ (which we identify with $V$) such that the success probability of our algorithm is equal to the success probability of an algorithm using input state $|\psi \rangle \in V \subset Y^{\uparrow}$ and the projective measurement $\{tEt^{-1}\}_t$ where $E$ denotes orthogonal projection onto $e \otimes Y$ and $t$ ranges over a set of coset representatives for $H$.

Now decompose $Y$ into irreducible $H$-invariant orthogonal subspaces:
$$
Y = Y_1 \oplus \dots \oplus Y_r.
$$
Then $Y^{\uparrow} \cong \bigoplus_i Y_i^{\uparrow}$ as $\C G$-modules. Let $\Pi_i$ denote orthogonal projection onto $Y_i^{\uparrow}$. Then $|\psi \rangle$ can be decomposed as a combination of orthogonal unit vectors
$$
|\psi \rangle = \lambda_1 |\psi_1 \rangle + \dots + \lambda_r |\psi_r \rangle
$$
such that each $|\psi_i \rangle$ belongs to $Y_i^{\uparrow}$. Even more is true: since $\lambda_i |\psi_i \rangle = \Pi_i |\psi \rangle$ and $\Pi_i$ is a $\C G$-module map, we know $|\psi_i \rangle \in (Y_i^{\uparrow})_V$.

Note also that $E$ decomposes as $E = E_1 + \dots + E_r$ where $E_i$ is orthogonal projection onto $e \otimes Y_i$.

We are ready to bound the success probability of the algorithm. Recall that we are using the measurement $\{tEt^{-1}\}_t$ to distinguish the density operators $\{t \rho t^{-1}\}_t$ where $\rho = \frac{1}{|H|}\sum_{h \in H} h |\psi \rangle \langle \psi | h^{-1}$. Then
$$
P_{\rm succ} = \frac{1}{|G:H|}\sum_t \Tr((t \rho t^{-1}) t E t^{-1}) = \langle \psi | E | \psi \rangle.
$$
Now using the decomposition of $|\psi \rangle$ we have
$$
\langle \psi | E |\psi \rangle = \sum_{i=1}^r |\lambda_i|^2 \langle \psi_i | E_i |\psi_i \rangle.
$$
Now by Proposition \ref{inducedProp2} we have, for all $i$,
$$
\langle \psi_i | E_i |\psi_i \rangle \leq \frac{\dim (Y_i^{\uparrow})_V}{\dim Y_i^{\uparrow}}.
$$
Therefore
$$
P_{\rm succ} \leq \sum_i |\lambda_i|^2 \frac{\dim (Y_i^{\uparrow})_V}{\dim Y_i^{\uparrow}} \leq \max_{Y \in \Irr(H)} \frac{\dim Y^{\uparrow}_V}{\dim Y^{\uparrow}}.
$$ \end{proof}

\subsection{Query complexity}
We now know the success probability of an optimal single-query algorithm solving coset identification. As in Section 4, we combine this with the fact that an optimal $t$-query algorithm with access to the representation $V$ is the same as an optimal $1$-query algorithm to $V^{\otimes t}$ to determine the optimal success probability for $t$-query algorithms:

\begin{cor}\label{cor-tqueryopt} Let $(G, V, \pi, f)$ describe a case of coset identification. Then an optimal $t$-query algorithm succeeds with probability
$$
P_{\rm opt} = \max_{Y \in \Irr(H)} \frac{\dim Y^{\uparrow}_{V^{\otimes t}}}{\dim Y^{\uparrow}}.
$$
\end{cor}

A straightforward consequence is the following:

\begin{thm}
Let $(G, V, \pi, f)$ describe a case of coset identification. Then the zero-error quantum query complexity of the problem is the minimum $t$ for which there exists some $Y \in \Irr(H)$ such that every irrep of $G$ appearing in $Y^{\uparrow}$ also appears in $V^{\otimes t}$.

The bounded error quantum query complexity is the minimum $t$ for which
$$
\max_{Y \in \Irr(H)} \frac{\dim (Y^{\uparrow})_{V^{\otimes t}}}{\dim Y^{\uparrow}} \geq 2/3.
$$
\end{thm}

\section{New examples of coset identification}
\label{sec:6}
\subsection{Identifying the coset of the Klein 4 group}
Here we present an easy demonstration of the machinery of the previous section. Consider the symmetric group on $4$ letters $G = S_4$ with normal subgroup the Klein $4$-group $H = \{e, (12)(34), (13)(24), (14)(23)\}$. Given access to the defining permutation representation $V$ of $S_4$ we would like to identify which coset of $H$ our permutation belongs to. Classically this requires $2$ queries. To determine the quantum complexity we need to know the characters of $V$ and $S_4$. Of course $V$ is isomorphic to $\Z_2 \times \Z_2$ (say, using the generators $(12)(34)$ and $(13)(24)$) and has $4$ characters labelled $\psi_{\alpha, \beta}$ with $\alpha, \beta \in \{0, 1\}$. The group $S_4$ has $5$ characters parametrized by partitions of $4$, denoted $\chi_{[4]}, \chi_{[3,1]}, \chi_{[2^2]}, \chi_{[2, 1^2]}$ and $\chi_{[1^4]}$. The restriction/induction rules are conveniently described in a {\it Bratteli diagram} (Figure \ref{bratteli}). \\

\begin{figure}[h!]
\centering
\includegraphics[scale=.8]{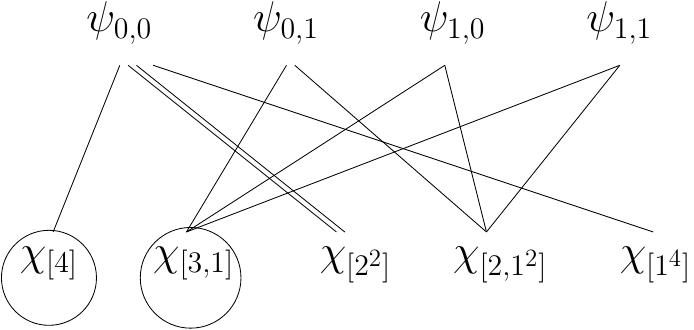}
\caption{Restriction/induction rules for $H < S_4$. The irreps appearing in $V$ are circled.}\label{bratteli}
\end{figure} 

The diagram indicates, for instance, that $\chi_{[2,1^2]}^{\downarrow} = \psi_{0,1} + \psi_{1,0} + \psi_{1,1}$ and $\psi_{0,0}^{\uparrow} = \chi_{[4]} + 2\chi_{[2^2]} + \chi_{[1^4]}$. Finally, we are given access to the defining permutation representation of $S_4$ which decomposes as $V = \chi_{[4]} + \chi_{[3,1]}$.

To find the optimal success probability of a single-query algorithm to determine which coset of $H$ a permutation belongs to, we examine the irreps of $H$ appearing in $V$. From the diagram we see that every irrep of $H$ appears in $V$, so we look at each one. First consider the trivial representation $\psi_{0,0}$. The only irrep of $S_4$ that appears in both $V$ and $\psi_{0,0}^{\uparrow}$ is $\chi_{[4]}$, which contributes a one dimensional subspace to the 6 dimensional $\psi_{0,0}^{\uparrow}$. Therefore using the irrep $\psi_{0,0}$ gives a success probability of $1/6$. Now consider $\psi_{0,1}$. In this case only $\chi_{[3,1]}$ appears in both $V$ and $\psi_{0,1}^{\uparrow}$, and it contributes $3$ dimensions to the 6 dimensional $\psi_{0,1}^{\uparrow}$. Therefore the success probability using this irrep is $3/6 = 1/2$. The other characters $\psi_{1,0}$ and $\psi_{1,1}$ give the same ratio so the optimal success probability of a single-query quantum algorithm is $1/2$ (note a single-query classical algorithm can do no better than probability $1/6$). 

That the optimal 2-query success probability is 1 can be verified using the fact that $V^{\otimes 2}$ contains a copy of every irrep of $S_4$ except the sign representation, and so using any of the irreps $\psi_{0,1}, \psi_{1,0}, \psi_{1,1}$ we can achieve probability 1.
\subsection{An action of the Heisenberg Group}

We now consider a natural action of the Heisenberg group over a finite field for which the oracle identification problem achieves a significant quantum speedup over the best classical algorithm. For this action, we also show that a single query suffices to solve the coset identification problem, where the chosen subgroup $H$ is the center of the group.

Specifically, let $p$ be prime and let $n$ be a positive integer.   Let $G=G(p,n)$ denote the Heisenberg group of all $(n+2)$-by-$(n+2)$ matrices with entries in $\mathbb{Z}_p$, $1$'s on the main diagonal and whose only other nonzero entries are in the first row and last column.  Such matrices are in correspondence with triples $(x,y,z)$, with $x,y\in \Z_p^n$ and $z\in \Z_p$, where $(1,x,z)$ is the first row of the matrix and $(z,y,1)$ is the last column of the matrix.  Then $G(p,n)$ is a $p$-group of order $p^{2n+1}$.

We consider the usual action of $G(p,n)$ on the set $X=\Z_p^{n+2}$, considered as column vectors, by matrix-vector multiplication.  The corresponding classical oracle identification problem turns out to have complexity $b(G)=n+1$.  To see this note that $y$ and $z$ can be determined by the single query $(0,\dots,0,1)$.  Further queries give affine conditions on $x$, and it requires at least $n$ of these to determine the value of $x$.

In contrast to the $n+1$ queries needed to solve this question classically, we now show that a single quantum query suffices to solve the problem with high probability, and that two queries suffice to solve the problem with certainty.

\begin{thm}\label{HeisOIP} Let $G(p,n)$ denote the Heisenberg group defined above acting by multiplication on the set of column vectors $X=\Z_p^{n+2}$.  Then an optimal single query quantum algorithm solves the oracle identification problem with probability
$$
P_{\rm opt} =1-\frac{1}{p} +\frac{2}{p^{n+1}}-\frac{1}{p^{2n+1}}.
$$
Furthermore, two queries suffices to solve the oracle identification problem with probability 1.
\end{thm}

We will prove this theorem shortly.  Before doing so, let us consider a related coset identification problem. Let $H<G(p,n)$ be the subgroup in which $x=y=0$. Then $H$ is a subgroup of order $p$, and in fact $H$ is the center of $G(p,n)$.  The coset identification problem with respect to this subgroup $H$ asks us to determine the values of $x$ and $y$.  In the classical case, $n+1$ queries are again required. However this time, a single quantum query solves the coset identification problem with certainty.

\begin{thm}\label{HeisCIP} Let $G=G(p,n)$ denote the Heisenberg group acting by multiplication on the set of column vectors $X=\Z_p^{n+2}$.  Let $H$ be center of $G$, the set of all matrices in $G$ for which $x=y=0$. Then the coset identification problem can be solved with a single quantum query with probability 1.
\end{thm}

In order to prove these theorems, we must understand the representation theory of $G=G(p,n)$, which we now describe briefly (for a concise and elegant review, see the letter by M. Isaacs to P. Diaconis published in the appendix to \cite{diaconis:2010}).
The group $G$ has $p^{2n}$ one-dimensional irreducible representations and $p-1$ irreducible representations 
of dimension $p^n$.  The one-dimensional representations will be denoted $\chi_{\alpha, \beta}$, indexed by tuples $\alpha, \beta\in \Z_p^n$. We identify these representations with their characters which are given by the formula
$$
\chi_{\alpha, \beta}(x,y,z) = \omega^{\alpha\cdot x+\beta\cdot y},
$$
with $\omega$ denoting a primitive $p$-th root of unity.

The $p^n$ dimensional representations denoted $\rho_c$, with $c\in\Z_p, c\neq 0$ are described as follows.     Let $U$ be the vector space of all complex-valued functions on $(\Z_p)^n$.   Fix  $c\in \Z_p$ with $c\neq 0$.  Then there is an irreducible representation 
$\rho_c$ of $G$ on $U$ given by
$$
[\rho_c(x,y,z)f](w) = \omega^{c(y\cdot w+z)} f(w+x).
$$
The character of this representation is given by
$$
\theta_c(x,y,z)=
\begin{cases}
p^n\omega^{cz} & \text{if }  x=y=0,\\
0 & \text{otherwise. }
\end{cases}
$$

In order to understand the query complexity of the oracle identification problem we must decompose the representation $V=\C^X$ into irreducible representations.  Since this representation comes from a permutation representation of $G$, each character value $\chi_V(x,y,z)$ is simply the number of fixed points of  the matrix $A=(x,y,z)$.  This number of fixed points is determined by the rank of the matrix
$A'= A - I$.  If $(x,y,z) =0$, then $A'$ has rank 0, and if $x$ and $y$ are both nonzero, then $A'$ has rank $2$.  In all other cases $A'$ has rank $1$.  We thus obtain the following character values of our given permutation representation $V$:
$$
\chi_V(x,y,z)=\begin{cases}
p^{n+2} & \text{if } (x,y,z) = (0,0,0) \\
p^{n} & \text{if } x\neq0 \text{ and } y\neq 0\\
p^{n+1} & \text{otherwise.}
\end{cases}
$$

To find the number of copies of the trivial representation $\chi_{0,0}$ appearing in $\chi_V$, we simply average these values and obtain $\langle \chi_V, \chi_{0,0}\rangle= p^n + 2(p-1)$.

Now let $\phi$ be any nontrivial irreducible character of $G$.  We compute the number  $\langle \chi_V, \phi \rangle$ of copies of $\phi$ appearing in $\chi_V$ as follows
\begin{align*}
\langle \chi_V, \phi \rangle &= \frac1{|G|}\sum_{(x,y,z)\in G} \chi_V(x,y,z) \phi(x,y,z) = \frac1{|G|}\sum_{(x,y,z)\in G}(\chi_V(x,y,z) - p^n) \phi(x,y,z) \\
&=\frac1{|G|}\sum_{(x,y,z)'} (p^{n+1} - p^n) \phi(x,y,z) + (p^{n+2} - p^n) \phi(0,0,0) \\
&= \frac{p-1}{p^{n+1}} \biggl[ \sum_{(x,y,z)'} \phi(x,y,z) + (p+1)\phi(0,0,0)\biggr]
\end{align*}
where $(x,y,z)'$ indicates a sum over those $(x,y,z)$ such that $x=0$ or $y=0$, but $(x,y,z)\neq(0,0,0)$. In the first line, we used the fact that if $\phi$ is nontrivial then $0 = \langle \phi, \mathbbm{1} \rangle = 1/|G| \sum_{(x,y,z) \in G} \phi(x,y,z)$.

Taking $\phi=\theta_c$ in this formula, we conclude that $V$ contains $p-1$ copies of $\rho_c$.  Taking $\phi=\chi_{\alpha,\beta}$, we get
$$
\langle \chi_V, \chi_{\alpha,\beta} \rangle =
\begin{cases}
p-1 &\text{if } \alpha=0 \text{ or } \beta=0, \text{ but not both}\\
0 & \text{if } \alpha\neq0 \text{ and } \beta\neq 0.
\end{cases}
$$

We conclude that our $V$ contains copies of all irreducible representations of $G$ except the $\chi_{\alpha,\beta}$ for which both $\alpha$ and $\beta$ are nonzero.  The optimal single-query quantum success probability is thus given by
$$
P_{\rm{opt}}=\frac1{|G|}\biggl(|G| - \sum_{\alpha,\beta\neq 0}\ 1 \biggr) = 1-\frac{1}{p} +\frac{2}{p^{n+1}}-\frac{1}{p^{2n+1}},
$$
as claimed.

If two queries are allowed, we have access to the representation $V\otimes V$.  Noting that $\chi_{\alpha, \beta}=\chi_{\alpha,0}\otimes\chi_{0,\beta}$, it follows that $V\otimes V$ contains every irreducible representation of $G$.  Hence, there is a probability 1 algorithm with two quantum queries.

Finally, we turn our attention to the coset identification problem for the subgroup $H=\{(0,0,z)|z\in \Z_p\}$.  To see that there is a probability one algorithm, note that any of the nontrivial characters of $H$ induces up to $p^n$ times one of the $\rho_c$. Since $\rho_c$ is contained in $V$, it follows that the coset identification problem can be solved with one query.
\qed

\subsection{Guessing the sign of a permutation}
Suppose there is an unknown permutation $g\in G=S_n$ for some $n\geq 2$. We wish to learn the sign of $g$ using queries to the standard action of $S_n$ on $\{1,..., n\}$. This is an instance of the hidden coset problem where $H=A_n$. Classically, $n-1$ queries are necessary to determine the sign of $g$.  In fact, any fewer queries and we do not learn anything about the sign.  Quantumly, we have

\begin{thm}\label{guessthesign} Let $n\geq 2$ and consider the standard action of $S_n$  on $\{1, \dots, n\}$. Consider the hidden coset problem for the subgroup $H=A_n$.  That is we wish to determine the sign of a hidden permutation.   For exact learning, $t=\lfloor \frac{n}{2} \rfloor$ quantum queries suffice.  With any smaller number of quantum queries, one cannot do any better than random guessing ($p=1/2$.)
\end{thm}

\begin{proof}
For facts and notation about representations of $S_n$ and $A_n$, we refer the reader to the proofs of Theorems \ref{guesstheperm} and \ref{guesstheeggplantperm}.

Let $V$ be the defining representation of $S_n$, and suppose we use $t$ queries so that we have access to $V'=V^{\otimes t}$.  Suppose $\lambda$ is a non-self-conjugate partition such that $V'$ contains  $V_{\lambda}$. Letting $Y=W_{\lambda}$, we see that $Y^{\uparrow}$ consists of one copy of $V_{\lambda}$ and one copy of $V_{\lambda^*}$.  Hence the quotient of dimensions
$$
\frac{\dim\ (Y^{\uparrow})_{V'}}{\dim\ Y^{\uparrow}}
$$
equals 1 if $V'$ contains both $V_{\lambda}$ and $V_{\lambda^*}$ and $\frac12$ if $V'$ contains $V_{\lambda}$ but not $V_{\lambda^*}$.  Now consider a self-conjugate partition $\lambda$ contained in $V'$. In this case, if we take $Y=W_{\lambda}^+$, then $Y^{\uparrow}$ is $V_{\lambda}$.  Hence in this case the quotient of dimensions is 1.
 
We thus wish to find the smallest $t$ such that $V^{\otimes t}$ contains both $V_{\lambda}$ and $V_{\lambda^*}$ for some partition $\lambda$ (including the possibility that $\lambda$ is self-conjugate).  For such $t$, we will have a $t$-query probability 1 algorithm and for fewer queries we cannot do better than probability 1/2, which is random guessing.

For even $n$, the value $t=n/2$ produces the partition $\lambda=(n/2 + 1, 1, \dots, 1)$ (with $n/2 -1$ 1's) and its conjugate $\lambda^*=(n/2, 1, \dots, 1)$ (with $n/2$ 1's).   For odd $n$, the value $t=\frac{n-1}{2}$ produces the self conjugate partition $(\frac{n+1}{2}, 1, \dots, 1)$ (with $\frac{n-1}{2}$ 1's).  In either case $t=\lfloor \frac{n}{2} \rfloor$ gives a probability 1 success, and fewer queries give success probability $1/2$.
\end{proof}

\section{Previously studied examples of coset identification} 
\label{sec:7}

Here we discuss the relation of this work to preceding work. To the authors knowledge, every previously studied special case of the general coset identification problem uses oracles sampled from an abelian group. Zhandry \cite{zhandry:2015} addresses this problem (calling it the {\it oracle classification problem}) and provides an expression for the optimal success probability essentially identical to \ref{CIThm}. Thus our results are a non-abelian generalization of Zhandry's work, which was a key inspiration for the present paper. We briefly explain why Zhandry's result is equivalent to ours and then examine some other more specialized and well-known problems.

Coset identification for an abelian group is described by a tuple $(A, V, \pi, f)$ with $A$ abelian and $f: A \to X$ distinct and constant on the cosets of a subgroup $B$. We remark that since $B$ is a normal subgroup it is possible to identify $X$ with the quotient group $A/B$ and $f$ with the standard homomorphism $A \to A/B$. Hence coset identification in this instance may also be called {\it homomorphism evaluation.} By Cor. \ref{cor-tqueryopt} the optimal success probability for a $t$-query algorithm to determine $f(a)$ is
$$
P_{\rm opt} = \max_{Y \in \Irr(B)} \frac{\dim Y^{\uparrow}_{V^{\otimes t}}}{\dim Y^{\uparrow}}.
$$
Since $B$ is abelian, $Y$ is 1-dimensional and $Y^{\uparrow}$ decomposes as $|A:B|$ many distinct $A$-characters (corresponding to the characters of $A/B$). Hence $\dim Y^{\uparrow}_{V^{\otimes t}}$ (which by definition is the dimension of the maximal subspace of $Y^{\uparrow}$ containing only characters in $V^{\otimes t}$) is exactly equal to the number of shared irreps, i.e. the cardinality of the set $I(Y^{\uparrow}) \cap I(V^{\otimes t})$. As $Y$ varies, these sets partition $I(V^{\otimes t})$ into equivalence classes $[\chi]$, and by Frobenius reciprocity two characters are equivalent if and only if their restrictions to $B$ are identical. Hence the equation above can be restated:

\begin{thm}\label{thm-zhandry}(Zhandry, (\cite{zhandry:2015}, Theorem 4.1))
The optimal success probability of a $t$-query algorithm for abelian coset identification is
$$
P_{\rm opt} = \frac{1}{|A:B|} \max_{\chi \in I(V^{\otimes t})} | [\chi] |.
$$
\end{thm}
Under this interpretation we're aiming to find the largest collection of characters appearing in $V^{\otimes t}$ which have the same restriction to $B$. Zhandry includes several nice applications of the previous theorem, explained in a linear algebraic framework. Below we readdress a couple of these problems (polynomial interpolation and group summation) using character theoretic language, and we revisit the van Dam algorithm \cite{vanDam:1998}.

\subsection{Polynomial interpolation} 

The polynomial interpolation problem as outlined by Zhandry \cite{zhandry:2015} and Childs, van Dam, Hung and Shparlinski \cite{CvDHS:2016} is as follows.  Let $F=\F_q$ where $q=p^r$ for some prime $p$.  Suppose we have an unknown polynomial $f(X)$ over $F$ of degree less than or equal to $d$ and we wish to determine $f$ using queries that provide the value $f(x)$ for $x\in F$.  That is access to $f$ is provided via the oracle $U_f$ acting on $V=\C^F\otimes \C^F$ by
$$
U_f:|x, s\rangle \mapsto  |x, s+f(x)\rangle.
$$
This equation defines a representation on $V$ of the group $G$ of all polynomials of degree less than or equal to $d$ under addition. 

We would like to see which of the characters of $G$ appear in this representation. Let $\omega$ be a primitive $p$-th root of unity and let $\Tr$ denote the trace map from $\F_q$ to $\F_p$. The characters of the additive group $F$ are given by $\chi_y$ with $y\in F$ defined by
$$
\chi_y(x) = \omega^{\Tr(yx)}.
$$
For $y\in F$, define the character state
$$
|\omega_y\rangle = \sum_{s\in F} \chi_y(-s)|s\rangle. 
$$
It is easy to see that
$$
U_f|x,\omega_y\rangle = \chi_y(f(x)) |x,\omega_y\rangle. 
$$
Thus if we let $V_{x,y}$ denote the 1-dimensional space spanned by $|x,\omega_y\rangle$, we have the decomposition
$$
V= \bigoplus V_{x,y}
$$
into irreducible representations.

The characters of $F^{d+1}$, which is isomorphic to $G$, are given, for $a\in F^{d+1}$, by  $\phi_a$, where
$$
\phi_a(c) = \omega^{\Tr(a\cdot c)}
$$
The character of $V_{x,y}$ is  $\phi_a$, with 
\begin{equation}\label{eq.yyx}
    a = (y, yx, yx^2, \dots, yx^d).
\end{equation}
To see this note that if $f(x)= \sum c_i X^i$, then
$$
\chi_y(f(x)) = \omega^{\Tr(yf(x))}= \omega^{c\cdot (y, yx, yx^2, \dots, yx^d) }
$$
Thus the irreps that appear in $V$ are exactly the $\phi_a$, where $a$ has the form in Equation \ref{eq.yyx}.  Since $\phi_a\otimes \phi_{a'} =\phi_{a+a'}$, it follows that the $k$-fold tensor power contains those $\phi_b$ where $b$ can be expressed as a $k$-fold sum of vectors of the form in Equation \ref{eq.yyx}.  This is exactly in image of the map $Z$ as described by Childs, van Dam, Hung and Shparlinski \cite{CvDHS:2016}, so we have reproved their Theorem 1.

The computation of the optimal success probability is now reduced to an algebraic/combinatorial problem which is nontrivial to solve (and is achieved in \cite{CvDHS:2016}). Hence this example serves to show the limitations of our main results: they can be used to translate questions about query complexity into purely algebraic problems which may or may not be easily solvable. The character theoretic technique shown above could also be used to reduce the query complexity of multivariable polynomial interpolation to a counting problem, as was achieved by Chen, Childs and Hung \cite{CCH:2018} without referring to characters. So far though, the character based language has not led to any progress on this problem.

\subsection{Group summation problem}
Fix an abelian group $G$. The $k$-element group summation problem is the task of computing the sum $f(1) + \dots + f(k)$ given access to an evaluation oracle hiding a function $f: \{1,2, \dots k\} \to G$.

This is an instance of coset identification. The oracles form a representation of the group of functions $\Fun([k], G) = \{ f: \{1, \dots, k\} \to G\}$, which we identify with $G^k$. In the quantum version they act on the Hilbert space $V = \C^k \otimes \C G$ via
$$
U_{f}|i, b \rangle = |i, b + f(i)\rangle.
$$
We wish to determine $\Sigma(f) := \sum_{i=1}^k f(i)$, which is the same as determining the coset of $f$ w.r.t. the subgroup
$$
H = \{f: \Sigma(f) = 0\} \leq \Fun([k], G).
$$
The irreducible characters of $\Fun([k], G)$ are all of the form $\chi_1 \times \dots \times \chi_k: G^k \to \C$, where each $\chi_i \in \Irr(G)$. The Hamming weight of such a character, denoted $\wt(\chi_1 \times \dots \times \chi_k)$, is the number of components which are nontrivial. The characters appearing in the evaluation representation on $V = \C^k \otimes \C G$ are exactly those with Hamming weight $\leq 1$. This implies that the characters appearing in $V^{\otimes t}$ are those with Hamming weight $\leq t$.

Next we consider the irreps of the subgroup $H$. We may $H$ with $G^{k-1}$ via
$$
(a_1, \dots, a_{k-1}, -(a_1 + \dots + a_{k-1})) \leftrightarrow (a_1, \dots, a_{k-1}).
$$
Hence irreps of $H$ may be written as $\tau_1 \times \dots \times \tau_{k-1}$ where again the $\tau_i$ are irreps of $G$. Using the above equation one verifies that two irreps $\chi_1 \times \dots \times \chi_k$ and $\eta_1 \times \dots \times \eta_k$ have the same restriction to $H$ if and only if there exists $\psi \in \Irr(G)$ such that
$$
\chi_1 \times \dots \times \chi_k = \psi \eta_1 \times \dots \times \psi \eta_k.
$$
By Zhandry's theorem (Thm. \ref{thm-zhandry}) the optimal success probability for a $t$-query algorithm is obtained by finding the largest collection of characters in $V^{\otimes t}$ which restrict to the same irrep of $H$. We can describe an element of such a maximal equivalence class: the character $\chi_1 \times \dots \times \chi_k$ should have at least $k-t$ trivial components (to guarantee its Hamming weight is $\leq t$), then $k-t$ components equal to some nontrivial character $\psi_1$ (so then $\chi_1 \psi_1^{-1} \times \dots \times \chi_k \psi^{-1}$ also has Hamming weight $\leq t$), another $k-t$ components equal to $\psi_2$, and so on. For instance, we may pick
$$
\chi = \mathbbm{1} \times \mathbbm{1} \times \dots \times \psi_1 \times \psi_1 \times \dots \times \psi_N \times \psi_N \dots
$$
where the characters $\mathbbm{1}, \psi_1, \dots, \psi_N$ are distinct (but otherwise arbitrary), and each one appears at least $k-t$ many times. Then the equivalence class of $\chi$ has  size $N+1$, consisting of the characters
\begin{align*}
[\chi] = \{&\chi, \\ 
&\psi_1^{-1} \times \psi_1^{-1} \dots \times \mathbbm{1} \times \mathbbm{1} \times \dots \times \psi_1^{-1} \psi_N \times \psi_1^{-1} \psi_N \times \dots, \\
&\vdots \\
&\psi_N^{-1} \times \psi_N^{-1} \times \dots \times \psi_N^{-1} \psi_1 \times \psi_N^{-1} \psi_1 \times \dots \times \mathbbm{1} \times \mathbbm{1} \times \dots\}
\end{align*}

The size $N+1$ of this equivalence class is either $|G|$ (if we can fit every irrep of $G$, which happens iff $\lfloor \frac{k}{k-t} \rfloor \geq |G|$) or $\lfloor \frac{k}{k-t} \rfloor$. Hence for a $t$-query algorithm
$$
P_{\rm opt} = \frac{1}{|G|} \min \left( \left\lfloor \frac{k}{k-t} \right\rfloor, |G| \right).
$$

This is exactly Thm. 5.1 by Zhandry (\cite{zhandry:2015}). An efficient algorithm achieving this success probability had previously been described (for $G$ cyclic) by Meyer and Pommersheim \cite{MePo:2011}.

\subsection{The van Dam algorithm}
The van Dam learning problem \cite{vanDam:1998} is concerned with identifying a (total) Boolean function $f: \{1, \dots, n\} \to \mathbb{Z}_2$ given access to evaluation queries. This is a special case of symmetric oracle discrimination (see Section 4). The group of oracles is isomorphic to $\mathbb{Z}_2^n$ and irreps can be again written as a product $\chi_1 \times \dots \times \chi_n$ of characters of $\mathbb{Z}_2$. The characters appearing in the $t$-th tensor power of the evaluation oracle representation are exactly those with Hamming weight $\leq t$. Hence the optimal success probability of a $t$-query algorithm is
$$
P_{\rm opt} = \frac{1}{2^n} \left| \{ \text{characters of $\mathbb{Z}_2^n$ with $\wt\ \leq t$ } \} \right| = \frac{1}{2^n} \sum_{i=0}^t {n \choose i}
$$
which reproves the optimality of van Dam's algorithm.

\section*{Acknowledgements} We would like to thank Andrew Childs, Hanspeter Kraft, David Meyer, Marino Romero and Mark Zhandry for helpful communications. We are grateful to the reviewers for their useful comments which have given us direction towards future work.

\bibliographystyle{alphaurl}
\bibliography{references}
\end{document}